\documentclass[12pt]{llncs}

\usepackage[T1]{fontenc}
\usepackage{amsfonts}
\usepackage{amssymb}
\usepackage{amsmath}
\usepackage{tikz}
\usetikzlibrary{automata}
\usetikzlibrary{shadows}
\usetikzlibrary{arrows}
\usetikzlibrary{shapes}
\usetikzlibrary{decorations.pathmorphing}
\usetikzlibrary{fit}
\usepackage{algorithm}
\usepackage{algorithmic}
\usepackage[margin=2cm]{geometry}
\newtheorem{clam}{Claim}
\usepackage{hyperref}

\def\slash{\relax\ifmmode\delimiter"502F30E\mathopen{}\else\@old@slash\fi}

\tikzstyle{every picture}=[
  >=stealth', shorten >=1pt, node distance=1.44cm,auto,bend angle=45,initial text=,
  every state/.style={inner sep=0.75mm, minimum size=1mm},font=\scriptsize,
]

\newcommand{\cqfd}{\hfill{\vrule height 3pt width 5pt depth 2pt}}

\def\sc{\mathrm{sc}}
\def\stararrow{\overset{*}{\rightarrow}}
\title{State complexity of catenation combined with a boolean operation: a unified approach}
\author{
    Pascal Caron
    \and Jean-Gabriel Luque
    \and Ludovic Mignot
    \and Bruno Patrou
    \thanks{\{Pascal.Caron, Jean-Gabriel.Luque, Ludovic.Mignot, Bruno.Patrou\}@univ-rouen.fr}
}
\institute{LITIS, Université de Rouen,\\ Avenue de l'Université,\\ 76801 Saint-\'Etienne du Rouvray Cedex,\\ France}

\begin{document}

\maketitle
\begin{abstract}
In this paper we study the state complexity of catenation combined with symmetric difference. First, an upper bound is computed using some combinatoric tools. Then, this bound is shown to be tight by giving a witness for it. Moreover, we relate this work with the study of state complexity for two other combinations: catenation with union and catenation with intersection. And we extract a unified approach which allows to obtain the state complexity of any combination involving catenation and a binary boolean operation.
\end{abstract}

\section{Introduction}

The study of state complexity is a very active research area for about 20 years. Its starting point is usually dated from 1994 with \cite{YZS94} although some early related works can be cited (especially \cite{Brz64,Mas70,Bir92}). Initially, was studied the state complexity of individual operations (see, for example, \cite{Jir05,JJS05} and the survey \cite{GMRY12}), then the complexity of combined operations was investigated since 2007. Initiated by \cite{SSY07}, it has resulted in numerous articles (\cite{GSY08,JO11,LMSY08}, for example). Among them, we pay special attention to \cite{CGKY11} where is studied the combination of catenation with union and intersection.

In their paper, the authors observe that "the state complexity of a combined operation is not simply a mathematical composition of the state complexities of its component operations". Indeed, the state complexity of a regular $n$-ary operation is a function giving, from $n$ integers standing for the sizes of minimal and complete DFAs, the maximal states number of a minimal and complete DFA accepting the resulting language of the operation applied over $n$ languages recognized by the DFAs whose sizes are given as inputs. And, not surprisingly, when combining two operations, the language obtained to reach the state complexity of the first applied operation is not necessarily a relevant input to reach the state complexity of the second operation.

The state complexity of the catenation of a $m$-states DFA with a $n$-states DFA is $(m-1)2^n+2^{n-1}$ (\cite{YZS94}) and the state complexity of any binary, non trivial, boolean operation of a $n$-states DFA with a $p$-states DFA is $np$ (\cite{GMRY12}). In \cite{CGKY11} it is proven that the state complexity of catenation combined with intersection corresponds to the mathematical combination of these complexities: $(m-1)2^{np}+2^{np-1}$ whereas, when catenation is combined with union, the state complexity falls down to $(m-1)(2^{n+p}-2^n-2^p+2)+2^{n+p-2}$. And the authors conclude their paper with the following question: "Why are the state complexity results on these two very similar combined operations so different ?" They also mentioned that "although there is only a limited number of individual operations, the number of combined operations is unlimited".

In this paper we provide some answers to these remarks and interrogations. Indeed, we will see that the combination $L_1(L_2\circ L_3)$ of catenation with any regular, non trivial, binary operation is equivalent to the combination $L_1(L\mbox{ op }L')$ of catenation with either $\cup$, either $\cap$ or $\oplus$ applied to languages chosen among $L_2$, $L_3$ and their complements. That is, it is sufficient to know the state complexities of catenation combined with these three binary operations to obtain all other cases. Moreover, while studying the only missing combination of the three (catenation with symmetric difference), it will emerge a unifying approach for the calculus of the three state complexities giving, by the way, a satisfactory answer to the question concerning the so different obtained complexities.

The article mainly focuses on calculus of the state complexity for the combination of catenation with symmetric difference. In particular, it will need some non trivial combinatorial tools in addition to the classical automata handling. The next section gives some definitions and notations about automata and combinatoric. 
Section~\ref{sec prelim}
contains some specific constructions needed later and describes more precisely how our approach meets, in 
standardizing them, the works done for catenation combined with union and intersection. In 
Section~\ref{upper},
we give an upper bound for the state complexity of $L_1(L_2\oplus L_3)$, assuming the number of some particular objects can be computed. 
In 
Section~\ref{se count tab}, 
we show how to compute these objects, curiously connecting our work to well known numbers in combinatorial area. In 
Section~\ref{sec wit},
 we build a witness proving the tightness of the bound given in 
Section~\ref{upper}. To achieve this work, we use the family of languages defined by Brzozowski in \cite{Brz13}. 
\section{Background}\label{background}
In all this paper, $\Sigma$ 
denote a finite alphabet. The set of all finite words over $\Sigma$ is denoted by $\Sigma ^*$.  The empty word is denoted by  $\varepsilon$. A language is a subset of $\Sigma^*$. The set of subsets of a finite set $A$ is denoted by $2^A$ and $\#A$ denotes the cardinality of $A$.  The symmetric difference is denoted by the symbol $\oplus$. We denote by $\uplus$ the union of disjoint sets. The symbol $\circ$ denotes any boolean operation on languages. In the following, by abuse of notation, we 
often write $q$ for any singleton $\{q\}$.

A  finite automaton (FA) is a $5$-tuple $A=(\Sigma,Q,I,F,\cdot)$ where $\Sigma$ is the input alphabet, $Q$ is a finite set of states, $I\subset Q$ is the set of initial states, $F\subset Q$ is the set of final states and $\cdot$ is the transition function from  $Q\times \Sigma$ to $2^Q$. A FA is deterministic (DFA) if $\#I=1$ and for all $q\in Q$, for all $a\in \Sigma$, $\#(q\cdot a)\leq 1$. 
 Let $a$ be a symbol of $\Sigma$. Let $w$ be a word of $\Sigma^*$. The transition function is  extended to any word  by $q\cdot a w=\bigcup _{q'\in q\cdot a} q'\cdot w$ and  $q\cdot \varepsilon=q$. 
  
A symmetric use of the dot notation leads  to the following definition.  Let $w\cdot q=\{q'\mid q\in q'\cdot w\}$.  We  extend the dot notation to any set of states $S$ by $S\cdot w=\bigcup_{s\in S}s\cdot w$ and $w\cdot S =\bigcup_{s\in S} w\cdot s$. 
 A word $w\in \Sigma ^*$ labels a successful path in a FA $A$ if $I\cdot w\cap F\neq \emptyset$. 

 In this paper, we assume that all FA  are complete which means that for all $q\in Q$, for all $a\in \Sigma$, $\#(q\cdot a)\geq 1$. A state $q$ is accessible in a FA  if there exists a word $w\in \Sigma ^*$ such that $q\in I\cdot w$.
The language recognized by a FA $A$ is the set of words labeling a successful path in $A$. 
Two automata are said to be equivalent if they recognize the same language.  

Let $D=(\Sigma,Q_D,i_D,F_D,\cdot)$ be a DFA.
Two states $q_1,q_2$ of  $D$ are equivalent if for any word $w$ of $\Sigma^*$, $q_1\cdot w\in F_D$ if and only if $q_2\cdot w\in F_D$. Such an equivalence is denoted by $q_1\sim q_2$. A DFA is  minimal if there does not exist any equivalent DFA  with less states and it is well known that for any DFA, there exists a unique minimal equivalent one \cite{HU79}. Such a minimal DFA  can be  obtained from $D$ by computing the accessible part of the automaton $D\slash \sim=(\Sigma,Q_D\slash \sim,[i_D],F_D\slash \sim,\cdot)$ where for any $q\in Q_D$, $[q]$ is the $\sim$-class of the state $q$ and for any $a\in \Sigma$, $[q]\cdot a=[q\cdot a]$. In a minimal DFA, any two distinct states are pairwise non-equivalent.

 Any nondeterministic finite automaton $B=(\Sigma,Q,I,F,\cdot)$ can be converted into an equivalent DFA $A=(\Sigma, Q',q'_0,F',\cdot)$ by a classical algorithm called the subset construction \cite{RS59}. The set of states is $Q'= 2^Q$, the initial state is $q'_0=I$,
and
the final states are the subsets of $Q$ containing a  state of $F$, that is  $F'=\{q'\in2^Q\mid q'\cap F\neq \emptyset\}$
.
 
   Two other classical constructions are the catenation $A\cdot B$  and the cartesian product $A\times B$ of two automata $A=(\Sigma,Q_A,I_A,F_A,\cdot_A)$ and $B=(\Sigma,Q_B,I_B,F_B,\cdot_B)$.  
  
  The catenation $A\cdot B$ is the automaton  $(\Sigma,Q=Q_A\cup Q_B,I,F,\cdot)$ defined by:
  \begin{itemize}
    \item
    $I=
        \left\{
          \begin{array}{l@{\ }l}
            I_A\cup I_B & \text{ if } I_A\cap F_A\neq \emptyset,\\
            I_A & \text{ otherwise,}\\
          \end{array}
        \right.$
    \item 
      $F=
        \left\{
          \begin{array}{l@{\ }l}
            F_A\cup F_B& \text{ if } I_B\cap F_B\neq \emptyset,\\
            F_B & \text{ otherwise,}\\
          \end{array}
        \right.$
    \item $\forall (q,a)\in Q\times\Sigma$, 
      $q \cdot a=
        \left\{
          \begin{array}{l@{\ }l}
            q \cdot_B a & \text{ if }q\in Q_B,\\
            q \cdot_A a \cup I_B & \text{ if } q \cdot_A a\cap F_A\neq\emptyset,\\ 
            q \cdot_A a & \text{ otherwise.}
          \end{array}
        \right.$
  \end{itemize}
  
  A cartesian product $A\times B$ is an automaton  $(\Sigma,Q=Q_A\times Q_B,I_A\times I_B,F,\cdot)$ such that
  $\forall ((q_1,q_2),a)\in Q\times\Sigma$, $ (q_1,q_2) \cdot a= q_1 \cdot_A a \times q_2 \cdot_B a$.
  Specifying $F$, we recover classical operations like the intersection ($F=F_A\times F_B$ in this case).

 The state complexity of a regular language $L$ denoted by $\sc(L)$ is the number of states of its minimal DFA. 
  Let ${\cal L}_n$ be the set of languages of state complexity $n$. The state complexity of a unary operation $\otimes$ is the function $\sc_{\otimes}$ associating with an integer $n$ the maximum of the state complexities of $(\otimes L)$ for $L\in {\cal L}_n$.
  A language $L\in {\cal L}_n$ is a witness (for $\otimes$) if  $\sc(\otimes L)=\sc_{\otimes}(n)$.
  This can be generalized, and the state complexity of a $k$-ary operation $\otimes$ is the $k$-ary function which associates with any tuple $(n_1,\ldots,n_k)$ the integer $\mathrm{max}\{\sc(\otimes(L_1,\ldots,L_k))|L_i\in\mathcal{L}_{n_i},\forall i\in[1,k]\}$. Then, a witness is a tuple $(L_{1},\ldots,L_{k})\in({\cal L}_{n_1}\times \cdots  \times{\cal L}_{n_k})$ such that $\sc(\otimes(L_{1},\ldots,L_{k}))=\sc_{\otimes}(n_1,\ldots,n_k)$. 
  An important research area consists in finding witnesses for any $(n_1,\ldots ,n_k)\in \mathbb{N}^k$
  and for any combination of elementary operations. Obviously, 
 
  \begin{clam}\label{claim1}
   The state complexity of an operation defined as a composition of more elementary ones is upper-bounded by the composition of the corresponding elementary state complexities. 
   \end{clam}

   For example, let us consider the ternary operation $\otimes$ defined for any three languages $L_1, L_2, L_3$ by $\otimes(L_1,L_2,L_3)=L_1\cdot(L_2\circ L_3)$ and let $h$ be its state complexity. 
  Let $f,g$ be the respective state complexities of $\cdot$ and $\circ$. For any three integers $n_1,n_2,n_3$, it holds $h(n_1,n_2,n_3)\leq f(n_1,g(n_2,n_3))$. Moreover, following~\cite{CGKY11}, if $\circ=\cap$ then $h(n_1,n_2,n_3)=f(n_1,g(n_2,n_3))$ whereas, $h(n_1,n_2,n_3)<f(n_1,g(n_2,n_3))$ when $\circ=\cup$.
%
%


 
In \cite{Brz13}, Brzozowski defines  a family of languages  that turns to be universal witnesses for several operations. The automata denoting these languages are called \textit{Brzozowski automata}.
We  need some background to define these automata. We 
follow the terminology of \cite{GM08}. Let $Q=\{0,\ldots, n-1\}$ be a set. A \textit{transformation} of the set $Q$ is a mapping of $Q$ into itself. If $t$ is a transformation and $i$ an element of $Q$, we denote by $it$ the image of $i$ under $t$. A transformation of $Q$ can be represented by $t=[i_0, i_1, \ldots i_{n-1}]$ which means that $i_k=kt$ for each $0\leq k\leq n-1$ and $i_k\in Q$. A \textit{permutation} is a bijective transformation on $Q$. The \textit{identity} permutation of $Q$ is denoted by $1_Q$. A \textit{cycle} of length $\ell\leq n$  is a permutation $c$, denoted   by $(i_0,i_1,\ldots i_{\ell-1})$, on a subset $I=\{i_0,\ldots ,i_{\ell-1}\}$ of $Q$  where  $i_kc=i_{k+1}$ for $0\leq k<\ell-1$ and $i_{\ell-1}c=i_0$.  A \emph{$k$-rotation}  is obtained by composing $k$ times the same cycle. In other word, we construct a $k$-rotation $r_k$ from the cycle $(i_0,\dots,i_{\ell-1})$ by setting  $i_{j}r_k=i_{j+k \mod \ell}$ for $0\leq j\leq \ell-1$. 
 A \emph{grouping} of two elements $i_j,i_k\in I$ is a $((k-j)\mod \ell)$-rotation   letting $i_k$ unchanged and obtained by iterating $(k-j)\mod \ell$ times the cycle $(i_0,\dots,i_{k-1},i_{k+1},\dots,i_{\ell-1})$. Such a grouping sends $i_j$ to $i_{k+1}$ and $i_k$ to $i_k$.
A \textit{transposition} $t=(i,j)$ is a permutation on $Q$ where $it=j$ and $jt=i$ and for every  elements $k\in Q\setminus \{i,j\}$, $kt=k$.  A \textit{contraction}  $t=\left(\begin{array}{r}i\\j\end{array}\right)$ is a transformation where  $it=j$ and  for every  elements $k\in Q\setminus \{i\}$, $kt=k$.

In any complete DFA $(\Sigma, Q, i, F, \cdot)$, any word of  $\Sigma^*$ induces a transformation over $Q$.
  Let $a,b,c,d$ be distinct  symbols of $\Sigma$. As an example of Brzozowski automata (see Figure \ref{Brzo}), let \label{Brzo-def} $W_n(a,b,c,d)=(\Sigma,Q_n,0,\{n-1\},\cdot)$ where $Q_n=\{0,1,\ldots ,n-1\}$, the symbol $a$ acts as the cycle  $(0,1,\ldots, n-1)$,  $b$ acts as the transposition $(n-2, n-1)$,  $c$ acts as the contraction $\left(\begin{array}{r}1\\0\end{array}\right)$ and   $d$ acts as $1_{Q_n}$. 

 \begin{figure}[htb]
	\centerline{
		\begin{tikzpicture}[node distance=1.5cm, bend angle=25]
			\node[state,initial] (0)  {$0$};
			\node[state] (1) [right of=0] {$1$};
			\node[state] (2) [right of=1] {$2$};
			\node (etc1) [right of=2] {$\ldots$};
			\node[state, rounded rectangle] (m-3) [right of=etc1] {$n-3$};
			\node[state, rounded rectangle] (m-2) [right of=m-3] {$n-2$};
			\node[state,accepting, rounded rectangle] (m-1) [right of=m-2] {$n-1$};
			\path[->]
        (0) edge[bend left] node {$a$} (1)
        (1) edge[bend left] node {$a$} (2)
        (2) edge[bend left] node {$a$} (etc1)
        (etc1) edge[bend left] node {$a$} (m-3)
        (m-3) edge[bend left] node {$a$} (m-2)
        (m-2) edge[bend left] node {$a, b$} (m-1)
        (m-1) edge[out=-115, in=-65, looseness=.2] node[above] {$a$} (0)
		    (0) edge[out=115,in=65,loop] node {$b, c, d$} (0)
		    (1) edge[out=115,in=65,loop] node {$b, d$} (1)
		    (2) edge[out=115,in=65,loop] node {$b, c, d$} (2)
		    (m-3) edge[out=115,in=65,loop] node {$b, c, d$} (m-3)
		    (m-2) edge[out=115,in=65,loop] node {$c, d$} (m-2)
		    (m-1) edge[out=115,in=65,loop] node {$c, d$} (m-1)
        (1) edge[bend left] node[above] {$c$} (0)
        (m-1) edge[bend left] node[above] {$b$} (m-2)
			;
\end{tikzpicture}
}
\caption{The automaton $W_n(a,b,c,d)$}\label{Brzo}
\end{figure}
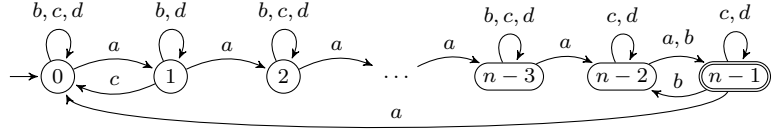	
Consider a discrete set of objects $\mathcal S$. The most general way to define formal   series is to consider families of complex numbers indexed by elements of $\mathcal S$,  $\mathcal A=(\alpha_{\mathfrak o})_{\mathfrak o\in\mathcal S}$, written formally as  infinite sums 
$\mathcal A\sim\sum_{\mathfrak o\in\mathcal S} \alpha_{\mathfrak o}\mathfrak o\in\mathbb C[[\mathcal S]],$
in order to endow the space  $\mathbb C[[\mathcal S]]$ with an algebraic structure related to a question that we want to study.

The notion of combinatorial classes (see \emph{e.g.} \cite{Flaj}) is suited for solving enumeration problems. Let us recall few basic definitions and properties. A combinatorial class is defined by a map $\alpha$ from a discrete set $\mathcal S$ (the objects we want to enumerate) to another discrete set $\mathcal T$ (the objects describing  a certain statistic on $\mathcal S$) such that each set $\mathcal S_\lambda=\{\mathfrak o\in\mathcal S\mid \alpha(\mathfrak o)=\lambda\}$
is finite for any $\lambda\in\mathcal T$. The series associated to $\alpha$ is
$
\mathrm{\sigma}_\alpha=\sum_{\lambda\in\mathcal T}\#\mathcal S_\lambda [\lambda]$.


The most classical examples of such a construction are the (univariate) generating series. In this case, we have $\mathcal T=\mathbb N$, $\alpha$ represents the size of the objects of $\mathcal S$ and the generating series is
$
\mathrm{\sigma}_\alpha=\sum_{n\in\mathbb N}\#\mathcal S_n [n].
$
In order to solve  enumeration problems, we have to endow the space of series with a product. Typically, most of the enumerations deal with one of the three products
\begin{itemize}
\item[$\bullet$] $[n]\cdot[m]=[n+m]$ (ordinary generating functions)
\item[$\bullet$] $[n]\cdot[m]=\binom{n+m}n[n+m]$ (exponential generating functions)
\item[$\bullet$] $[n]\cdot[m]=[nm]$ (Dirichlet generating functions)
\end{itemize}
according to our initial problem. Depending on the product, one chooses to map basic elements on the different functions:
\begin{itemize}
\item[$\bullet$] $[n]\sim x^n$ (ordinary generating functions)
\item[$\bullet$] $[n]\sim {x^n\over n!}$ (exponential generating functions)
\item[$\bullet$] $[n]\sim {1\over n^x}$ (Dirichlet generating functions).
\end{itemize}
The products on series implement convolution products on  sequences:
\begin{itemize}
\item[$\bullet$] Ordinary generating functions: $(\alpha_n)_{n\in\mathbb N}\star(\beta_n)_{n\in\mathbb N}=\left(\sum_{i+j=n}\alpha_i\beta_j\right)_{n\in\mathbb N}$ (ordinary convolution)
\item[$\bullet$] Exponential generating functions: 
\begin{equation}\label{binconv}(\alpha_n)_{n\in\mathbb N}\star(\beta_n)_{n\in\mathbb N}=\left(\sum_{i+j=n}\left(n\atop i\right)\alpha_i\beta_j\right)_{n\in\mathbb N}\end{equation}  (binomial convolution)
\item[$\bullet$] Exponential generating functions: $(\alpha_n)_{n\in\mathbb N\setminus\{0\}}\star(\beta_n)_{n\in\mathbb N\setminus\{0\}}=\left(\sum_{ij=n}\alpha_i\beta_j\right)_{n\in\mathbb N\setminus\{0\}}$ (Dirichlet convolution).
\end{itemize}

For instance, the binomial transform of a sequence is interpreted in terms of exponential generating series by  the multiplication by $\exp(x)$:
\begin{equation}\label{binomialtransform}
\exp(x)\sum_{n\geq 0}\alpha_n{x^n\over n!}=\sum_{n\geq 0}\left(\sum_{i=0}^n\binom ni\alpha_i\right){x^n\over n!}.
\end{equation}

At the other end, we define the characteristic series $\underline{\mathcal S}$ of a set $\mathcal{S}$ as 
$
\underline{\mathcal S}=\mathrm{\sigma}_{Id}$, \textit{i.e.} $\sum_{\mathfrak o\in\mathcal S}\mathfrak o
$
where $Id(\mathfrak o)=\mathfrak o$ denotes the identity on $\mathcal S$. In the sequel, we will use characteristic series of languages on an alphabet $X$; these series belong to the algebra $\mathbb C\langle\langle X\rangle\rangle$ which is endowed with the catenation product extended by distributivity and continuity.

Between these two extremes, many configurations are possible.
We consider the case where $\mathcal T=\mathbb N^{(X)}$, the set of the integers sequences with a finite support (\emph{i.e.} all the elements of the sequence are zero except a finite number) indexed by the set $X$ (the set of variables). The series are realized by commutative variables as follows:
\begin{itemize}
\item[$\bullet$] $v\sim \prod_{x\in X}x^{v[x]}$ (ordinary multivariate generating functions)
\item[$\bullet$] $v\sim \prod_{x\in X}{x^{v[x]}\over v[x]!}$ (exponential multivariate generating functions)
\item[$\bullet$] $v\sim \prod_{x\in X}\frac1{{v[x]}^x}$ (Dirichlet multivariate generating functions).
\end{itemize}

Alternatively, we can replace a vector $v\in\mathbb N^{(X)}$ by a partition of an integer. A partition of $n$ is a decreasing finite sequence $\lambda=[\lambda_1,\dots,\lambda_k]$ of strictly positive integers verifying $\lambda_1+\cdots+\lambda_k=n$; we will denote this by $\lambda\vdash n$. The number of parts of $\lambda$, $k$ in this context, is also denoted by $\#\lambda$. The product of two partitions $[\lambda_1].[\lambda_2]$ equals the partition obtained by sorting the catenation of the two partitions in decreasing order. For instance set $\lambda_1=[2,2,1]\vdash 5$ and $\lambda_2=[3,2,1]\vdash 6$ then $[\lambda_1].[\lambda_2]=[3,2,2,2,1,1]$. If $X=\{x_1,x_2,\dots\}$ is a discrete alphabet the algebra generated by the elements $[\lambda]$ is isomorphic to the algebra obtained by endowing $\mathbb N^{(X)}$ with the product $[v][v']=[v+v']$ (this corresponds to the ordinary multivariate generating functions). The isomorphism sends explicitly $\lambda$ to the vector $v$ such that $v[x_i]$ is the number of parts in $\lambda$ equal to $i$. For instance $[5,3,2,2,1,1]\sim [2,2,1,0,1,0,\ldots]$. So we can replace the statistic described by vectors of $\mathbb N^{(X)}$ by statistic described by partitions.

In the paper we will also manipulate set partitions of $\{1,\dots,n\}$. In the aim to avoid confusion we will denote (integer) partitions by the symbols $\lambda, \lambda_i, \lambda', \lambda'_i$ \emph{etc.} and the set partitions by the symbols $\pi, \pi_i, \pi', \pi'_i$ \emph{etc.} If $\pi$ is a set partition of a set $E$, we will denote $\pi\vDash E$. For any $\pi'\subset\pi$, we also denote $\pi'\subset_{\vDash}E$. When $E=\{1,\dots,n\}$, we will write $\pi\vDash n$ for $\pi\vDash E$ and  $\pi'\subset_{\vDash}n$ for  any subset $\pi'$ of $\pi$.


\section{Preliminaries}\label{sec prelim}

 There exist 16 binary boolean functions that can be expressed only with the operators $\wedge$, $\vee$ and $\Delta$ acting on variables and their negation (see Table~\ref{un tableau avec des fonctions booleenes binaires}). Each function is related to a set operation based on the classical correspondence $(\wedge,\vee,\neg,\Delta)\ \longleftrightarrow\ (\cap,\cup,\overline{\color{white}x},\oplus)$. For simplicity, we only use the second list of symbols to denote both.

  \begin{table}[H]
  \centerline{\begin{tabular}{p{0.45cm}p{0.45cm}p{0.45cm}p{0.45cm}p{0.45cm}p{0.45cm}p{0.45cm}p{0.45cm}p{0.45cm}p{0.45cm}p{0.45cm}p{0.45cm}p{0.45cm}p{0.45cm}p{0.45cm}p{0.45cm}}
   \rotatebox{45}{$\emptyset$}   & 
   \rotatebox{45}{$N \cap P$}  &
   \rotatebox{45}{$N \cap \overline{P}$}  & 
   \rotatebox{45}{$N$}  & 
   \rotatebox{45}{$ \overline{N} \cap  P$} & 
   \rotatebox{45}{$P$}  & 
   \rotatebox{45}{$N \oplus P$}  &
   \rotatebox{45}{$N \cup  P$}  & 
   \rotatebox{45}{$\overline{N} \cap  \overline{P}$}  & 
   \rotatebox{45}{$ \overline{N} \oplus P$}  &
   \rotatebox{45}{$ \overline{P}$}  & 
   \rotatebox{45}{$N \cup \overline{P}$}  & 
   \rotatebox{45}{$\overline{N}$}  & 
   \rotatebox{45}{$\overline{N} \cup  P$}  & 
   \rotatebox{45}{$\overline{N} \cup  \overline{P}$}  & 
   \rotatebox{45}{$\Sigma^*$}  \\
   $0$ & $0$ & $0$ & $0$ & $0$ & $0$ & $0$ & $0$ & $1$ & $1$ & $1$ & $1$ & $1$ & $1$ & $1$ & $1$\\
   $0$ & $0$ & $0$ & $0$ & $1$ & $1$ & $1$ & $1$ & $0$ & $0$ & $0$ & $0$ & $1$ & $1$ & $1$ & $1$\\
   $0$ & $0$ & $1$ & $1$ & $0$ & $0$ & $1$ & $1$ & $0$ & $0$ & $1$ & $1$ & $0$ & $0$ & $1$ & $1$\\
   $0$ & $1$ & $0$ & $1$ & $0$ & $1$ & $0$ & $1$ & $0$ & $1$ & $0$ & $1$ & $0$ & $1$ & $0$ & $1$\\
  \end{tabular}}
  \caption{The 16 binary boolean functions.}
  \label{un tableau avec des fonctions booleenes binaires}
  \end{table}

Let us consider three regular languages $M, N, P$ with respective deterministic state complexity $m, n, p\geq 3$. We look at the state complexity of $M(N\circ P)$ where $\circ$ is a boolean operation. 


Let $A=(\Sigma,Q_A=\{0,1,\ldots,m-1\},0,F_A,\cdot)$, $B=(\Sigma,Q_B=\{q_0,q_1,\ldots,q_{n-1}\},q_0,$ $F_B,\cdot)$, $C=(\Sigma,Q_C=\{r_0,r_1,$ $\ldots,r_{p-1}\},r_0,F_C,\cdot)$ be the three minimal DFAs for $M$, $N$, $P$ respectively and consider the  DFA $D=(\Sigma,Q_D=Q_A\times 2^{Q_B\times Q_C},i_D,F_D,\cdot)$ defined as follows: 
\begin{itemize}
\item $i_D=\left\{\begin{array}{ll}
(0,\emptyset)&\text{ if } 0\not\in F_A\\
(0,(q_0,r_0))&\text{otherwise}
\end{array}
\right.$ 
\item the set of final states 
$F_D$
depends on the operation $\circ$:
 %
 a state $(i,S)$ is final if $S$ contains a couple $(q,r)$ satisfying $(q\in F_B)\circ (r\in F_C)$ and
\item for a symbol $a\in \Sigma$ and a state $(i,S)$ in $D$, 
the function $\cdot$ is defined by

$(i,S)\cdot a=
  \left\{
    \begin{array}{ll}
      (i\cdot a,S\cdot a)&\text{ if }i\cdot a \notin  F_A \\
      (i\cdot a,S\cdot a\cup \{(q_0,r_0)\}) &\text{ otherwise,}
    \end{array}
  \right.$

  where the dot notation is extended to a set of couples by $S\cdot w=\bigcup_{(q,r)\in S}(q\cdot w)\times(r\cdot w)$ and $w\cdot S=\bigcup_{(q,r)\in S}(w\cdot q)\times(w\cdot r)$. 
\end{itemize}
 %
 
   The automaton $A$ being deterministic, the definition of $D$ implies that the accessible part of the subset automaton of $A\cdot (B\circ C)$ is  the accessible part of $D$, therefore  $L(D)=M(N\circ P)$.

For any state $(i,S)$, the set $S$ can be seen as a tableau with $n$ rows and $p$ columns where any cell $(j,k)$ is marked if and only if the couple of states $(q_j,r_k)$ is in $S$ (see Figure \ref{tableau}). In the following, for simplicity, when the dimensions are fixed, we 
assimilate a tableau with the set of its marked cells.

  \begin{figure}[H]
    \centerline{ 
      \begin{tikzpicture}[scale=0.5]   
	    \foreach \x in {1,...,7} {
	      \foreach \y in {1,...,6} {
	        \pgfmathparse{\x+1} \let\z\pgfmathresult
	        \pgfmathparse{\y+1} \let\t\pgfmathresult
	        \draw[fill=gray!40] (\x,\y) rectangle (\x+1,\y+1);
	      }
	    }  
	    \foreach \x/\y in {3/3,6/3,6/5} {
	        \pgfmathparse{\x+1} \let\z\pgfmathresult
	        \pgfmathparse{\y+1} \let\t\pgfmathresult
	        \draw[fill=white] (\x,\y) rectangle (\x+1,\y+1);
	        \draw (\x,\y) -- (\z,\t);
	        \draw (\z,\y) -- (\x,\t);
	    }
      \end{tikzpicture}
    }
    \caption{The tableau corresponding to $S=\{(q_3,r_2),(q_1,r_5),(q_3,r_5)\}$ with $n=6$ and $p=7$}
    \label{tableau}    
  \end{figure}
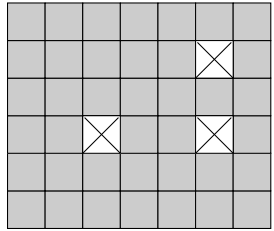 

Since the state complexity of catenation is $ \sc_\bullet(m,m')=(m-1)2^{m'}+2^{m'-1}$ and the state complexity of a binary boolean operation
$\circ$
is upperbounded by $\sc_\circ(n,p)=np$, from Claim \ref{claim1}, their composition allows to upperbound
the state complexity of  $M(N\circ P)$ by $(m-1)2^{np}+2^{np-1}$. This bound is reached when $\circ=\cap$ \cite{CGKY11}.
A study of $\circ=\cup$ has been done in the same paper, but the approach was not exactly the same. 

Indeed, they use the fact that the union $N \cup P$ can be performed without computing any cartesian product, but by computing the union of the states of the automata $B$ and $C$. Thus, the catenation of $A$ with $B\cup C$ is an automaton the determinization of which produces states of the form $(j,T,T')$, with $j$ a state of $A$ and  $T, T'$ two subsets of  states of $B$ and $C$ respectively. We can also notice that if $j$ is final, the sets $T$ and $T'$ must contain the initial state of $B$ and $C$ respectively, and therefore that $T$ is empty if and only if $T'$ is empty. Then the bound is $(m-1)(2^{n+p}-2^n-2^p+2)+2^{n+p-2}$ (only reached when $A$ has only one final state).

The state complexities of $M(N\cap P)$ and of $M(N\cup P)$ are known. Then, knowing the state complexity of $M(N\oplus P)$ will lead us to the knowledge of the state complexity of $M(N\circ P)$ for any non degenerative boolean operator. Indeed, as seen in Table \ref{un tableau avec des fonctions booleenes binaires}, there are $16$ boolean operations over two languages $N$ and $P$ and $6$ of them are degenerative $(\emptyset,\Sigma^*,N,P,\overline{N}, \overline{P})$. All the others can be expressed using one of  the three boolean operators $\cap$, $\cup$ and $\oplus$ and the complement operator over the single language $N$ or $P$. As the state complexity of the complement of $N$ (obtained by interchanging the states finality of $B$) is the state complexity of $N$,  the state complexity of $M(N\oplus P)$ will give us the state complexity of $M(N\circ P)$ for any boolean operator $\circ$.    Indeed, having a witness for $M(N\oplus P)$, $M(N\cup P)$ and $M(N\cap P)$ will give us a witness for any combination $M(N\circ P)$ as $N\circ P$ can be expressed through one of the three boolean operations applied to the complementation of $N$ and/or $P$. 
   For example suppose that the triple of automata $A, B, C$ is  a witness for the combination $M(N\cap P)$. We want to produce a witness for the combination $M(N-P)$. First, notice that this combination is $M(N\cap \overline{P})$. So the triple of automata $A, B$ and $\overline{C}$ is a witness for   $M(N-P)$ where   $\overline{C}$ is   the automaton $C$ in which states  finality  has been interchanged.


 The state complexity of $M(N\cup P)$ can be reinterpreted as follows in the automaton $D$. Let $(i,S)$ and $(i,S')$ be two distinct states such that the couples $(q_x,r_{x'})$ and $(q_y,r_{y'})$ are in $S'$ and $S=S'\cup \{(q_x,r_{y'})\}$. Then the two states $(i,S)$ and $(i,S')$ are  equivalent. Indeed,  to separate  these states, one has to find a word $w$ such that $(1)$ $S'\cdot w$ is equal to a set of couples which members are both non-final and $(2)$ $(q_x,r_{y'})\cdot w$ leads to a couple of states at least one of the two is final. The fact that  $q_x\cdot w$ or $r_{y'}\cdot w$ is final is contradicting $(1)$. So $(i,S)$ and $(i,S')$ are equivalent.

 Such equivalent states have tableaux with specific patterns. Indeed, the tableaux for $S$ and $S'$ contain the pattern of Figure \ref{fig two tab}(a) and Figure \ref{fig two tab}(b) respectively. None of them can be distinguish from the pattern of  Figure \ref{fig two tab}(c). So the number of equivalent states is the number of undistinguishable tableaux which are represented by the patterns of  Figure \ref{fig two tab}. The number of tableaux not containing patterns of  Figure \ref{fig two tab}(a) or Figure \ref{fig two tab}(b) is $(2^n-1)(2^p-1)+1$.
 
%
 \begin{figure}[htb]
 
	\begin{minipage}{0.33\linewidth}
	\centerline{
  \begin{tikzpicture}[scale=0.5]   
	    \foreach \x in {1,...,6} {
	      \foreach \y in {1,...,5} {
	        \pgfmathparse{\x+1} \let\z\pgfmathresult
	        \pgfmathparse{\y+1} \let\t\pgfmathresult
	        \draw[fill=gray!40] (\x,\y) rectangle (\x+1,\y+1);
	      }
	    }  
	    \foreach \x/\y in {2/2,5/4} {
	        \pgfmathparse{\x+1} \let\z\pgfmathresult
	        \pgfmathparse{\y+1} \let\t\pgfmathresult
	        \draw[fill=white] (\x,\y) rectangle (\x+1,\y+1);
	        \draw (\x,\y) -- (\z,\t);
	        \draw (\z,\y) -- (\x,\t);
	    }	   
      \end{tikzpicture}
}
\end{minipage}
	\begin{minipage}{0.33\linewidth}
	\centerline{
  \begin{tikzpicture}[scale=0.5]   
	    \foreach \x in {1,...,6} {
	      \foreach \y in {1,...,5} {
	        \pgfmathparse{\x+1} \let\z\pgfmathresult
	        \pgfmathparse{\y+1} \let\t\pgfmathresult
	        \draw[fill=gray!40] (\x,\y) rectangle (\x+1,\y+1);
	      }
	    }  
	    \foreach \x/\y in {2/2,5/2,5/4} {
	        \pgfmathparse{\x+1} \let\z\pgfmathresult
	        \pgfmathparse{\y+1} \let\t\pgfmathresult
	        \draw[fill=white] (\x,\y) rectangle (\x+1,\y+1);
	        \draw (\x,\y) -- (\z,\t);
	        \draw (\z,\y) -- (\x,\t);
	    }
      \end{tikzpicture}
}
\end{minipage}
	\begin{minipage}{0.33\linewidth}
	\centerline{
  \begin{tikzpicture}[scale=0.5]   
	    \foreach \x in {1,...,6} {
	      \foreach \y in {1,...,5} {
	        \pgfmathparse{\x+1} \let\z\pgfmathresult
	        \pgfmathparse{\y+1} \let\t\pgfmathresult
	        \draw[fill=gray!40] (\x,\y) rectangle (\x+1,\y+1);
	      }
	    }  
	    \foreach \x/\y in {2/2,5/2,2/4, 5/4} {
	        \pgfmathparse{\x+1} \let\z\pgfmathresult
	        \pgfmathparse{\y+1} \let\t\pgfmathresult
	        \draw[fill=white] (\x,\y) rectangle (\x+1,\y+1);
	        \draw (\x,\y) -- (\z,\t);
	        \draw (\z,\y) -- (\x,\t);
	    }
      \end{tikzpicture}
}
\end{minipage}
\caption{Three undistinguishable tableaux (a), (b), (c), for the union operator}\label{fig two tab}
\end{figure}
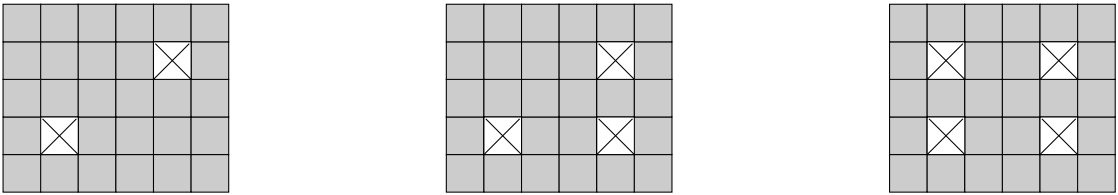	
	
 Indeed, one has to choose among $n$ rows and $p$ columns (at least one of each) and mark every cell at the intersection of the chosen rows and columns ($(2^n-1)(2^p-1)$) plus one configuration with no cell marked. We also have to count the same tableaux but with the cell $(0,0)$ marked ($2^{n-1}2^{p-1}$ tableaux). These observations 
lead
to the state complexity $(m-1)((2^n-1)(2^p-1)+1)+2^{n-1}2^{p-1}=(m-1)(2^{n+p}-2^n-2^p+2)+2^{n+p-2}$ of  $M(N\cup P)$.
 
  As there exist languages $M$, $N$ and $P$ such that there are no undistinguishable tableaux for $M(N\cap P)$, the state complexity  coincides with the bound. The study of undistinguishable tableaux for $M(N\oplus P)$ is done in the next section.

\section{An upper bound for the state complexity of $M(N\oplus P)$}\label{upper}

  In this section, we proceed as in the previous one by considering the same automata $A$, $B$, $C$ and $D$, setting $\circ=\oplus$.
  Notice that 
   a state $(i,S)$ of $D$ is final if and only if $S$ contains a couple $(q_x,r_y)$ such that either $q_x\in F_B$ or $r_y\in F_C$ but not both.
     
    As for the union (see the previous section), some particular states are necessarily equivalent. Let $(i,S)$ and $(i,S')$ be two distinct states such that the couples $(q_x,r_{x'})$, $(q_x,r_{y'})$ and $(q_y,r_{y'})$ are in $S'$ and $S=S'\cup \{(q_{y},r_{x'})\}$. Then the two states $(i,S)$ and $(i,S')$ are  equivalent.
  Indeed, if a word $w$ separates  $(i,S)$ and $(i,S')$, then $w$ sends $q_y$ in $F_B$ or $r_{x'}$ in $F_C$ but not both,  sending $(i,S)$ to a final state of $D$. This cannot be achieved without sending  $(i,S')$ to a final state of $D$, thus contradicting the separation by $w$. This is formally proved in the next lemma.
  
  \begin{lemma}\label{lem pas separ couples}
Let E=$\{q_x,q_{y}\}\times \{r_{x'},r_{y'}\}$. Let $(q,r)\in E$ and $w\in \Sigma^*$ with 
       $(q\cdot w\in F_B)\oplus (r\cdot w\in F_C)$
       Then there exists $(q',r')\in E$ such that  
       \begin{equation}\label{eq1}
       (q',r')\neq(q,r)\wedge (q'\cdot w\in F_B)\oplus(r'\cdot w\in F_C)
       \end{equation}
        \end{lemma}
  \begin{proof}
    Without loss of generality, let us set $(q,r)=(q_x,r_{x'})$. 
    We consider only the case where $q_x\cdot w\in F_B$ and $r_{x'}\cdot w\notin F_C$; the other case can be proved similarly. The couple 
    
    $(q',r')=\left\{\begin{array}{lll}
    (q_y,r_{x'}) &\mbox{if }q_y\cdot  w\in F_B\\
    (q_y,r_{y'}) &\mbox{if }q_{y}\cdot w \notin F_B\wedge r_{y'}\cdot w\in F_C\\
    (q_x,r_{y'}) &\mbox{otherwise}
    \end{array}\right.$
    
    satisfies (\ref{eq1}).
    \cqfd
  \end{proof}
  
  Such equivalent states imply undistinguishable tableaux as described below.
%
  Four distinct marked cells $s_1$, $s_2$, $s_3$ and $s_4$ define a \emph{rectangle} if there exist four integers $x$, $x'$, $y$ and $y'$ such that $\{s_1,s_2,s_3,s_4\}=\{q_x,q_{y}\}\times\{r_{x'},r_{y'}\}$. 
  Three distinct marked cells $s_1$, $s_2$ and $s_3$ form a \emph{right triangle} if there exists an unmarked cells $s_4$ such that $s_1$, $s_2$, $s_3$ and $s_4$ form a rectangle (See Figure~\ref{fig rect} and Figure~\ref{fig tri}).   A tableau $S$ is \emph{saturated} if it does not contain any right triangle. For each tableau $S$, we define $\mathrm{Sat}(S)$ as the smallest saturated tableau containing $S$. We can notice that $\mathrm{Sat}(S)$ is the intersection of all saturated tableaux containing $S$. Its existence is ensured since the tableau with each cell marked is saturated. Its unicity is due to the fact that the intersection of two saturated tableau containing $S$  is still a saturated tableau containing $S$.

  \begin{minipage}{0.45\linewidth}
    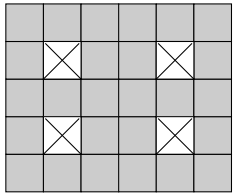
\begin{figure}[H]
      \centerline{
        \begin{tikzpicture}[scale=0.5]   
	      \foreach \x in {1,...,6} {
	        \foreach \y in {1,...,5} {
	          \pgfmathparse{\x+1} \let\z\pgfmathresult
	          \pgfmathparse{\y+1} \let\t\pgfmathresult
	          \draw[fill=gray!40] (\x,\y) rectangle (\x+1,\y+1);
	        }
	      }  
	      \foreach \x/\y in {2/2,5/4,2/4,5/2} {
	        \pgfmathparse{\x+1} \let\z\pgfmathresult
	        \pgfmathparse{\y+1} \let\t\pgfmathresult
	        \draw[fill=white] (\x,\y) rectangle (\x+1,\y+1);
	        \draw (\x,\y) -- (\z,\t);
	        \draw (\z,\y) -- (\x,\t);
	      }
        \end{tikzpicture}
      }
      \caption{A rectangle.}
      \label{fig rect}
    \end{figure}
  \end{minipage}
  \hfill  
  \begin{minipage}{0.45\linewidth}
    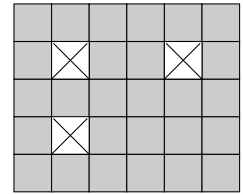
\begin{figure}[H]
      \centerline{
        \begin{tikzpicture}[scale=0.5]   
	      \foreach \x in {1,...,6} {
	        \foreach \y in {1,...,5} {
	          \pgfmathparse{\x+1} \let\z\pgfmathresult
	          \pgfmathparse{\y+1} \let\t\pgfmathresult
	          \draw[fill=gray!40] (\x,\y) rectangle (\x+1,\y+1);
	        }
	      }  
	      \foreach \x/\y in {2/2,5/4,2/4} {
	        \pgfmathparse{\x+1} \let\z\pgfmathresult
	        \pgfmathparse{\y+1} \let\t\pgfmathresult
	        \draw[fill=white] (\x,\y) rectangle (\x+1,\y+1);
	        \draw (\x,\y) -- (\z,\t);
	        \draw (\z,\y) -- (\x,\t);
	      }
        \end{tikzpicture}
      }
      \caption{A right triangle.}
      \label{fig tri}
    \end{figure}
  \end{minipage}
  
Alternatively, we can define saturated tableaux as the final step of a confluent rewriting process.
We denote by $S\rightarrow S'$ if $S'$ is obtained from $S$ by marking one cell which completes  a right triangle in $S$ into a rectangle in $S'$. The reflexive and transitive closure $\stararrow$ of $\rightarrow$ is a confluent partial order whose maximal elements are the saturated tableaux.
It follows that a property is preserved by the transformation $\stararrow$ if and only if it is preserved by the transformation $\rightarrow$.
We illustrate this principle with two lemmas which will be used below.

\begin{lemma}\label{lemma-distinct}
Let $S$ be a tableau where a cell $(j,k)$ is marked. If a tableau $S'$ contains no marked cell on row $j$ (equivalently on column $k$) then $\mathrm{Sat}(S)\neq\mathrm{Sat}(S')$.
\end{lemma}
\begin{proof}
First remark that if $T\rightarrow T'$ then $T$ contains no marked cell on row $j$ (resp. column $k$) if and only if $T'$ contains no marked cell on row $j$ (resp. column $k$). Indeed,
$T$ differs from $T'$  in only one marked cell which completes a right triangle in $T$ into a rectangle in $T'$. So, this cell belongs to a row and a column in $T$ which both contain a marked cell.

We deduce that if $T\stararrow T'$ then $T$ contains no marked cell on  row $j$ (resp. column $k$) if and only if $T'$ contains no marked cell on row $j$ (resp. column $k$).

Let $S$ and $S'$ be two  tableaux such that a cell $(j, k$) is marked in $S$ and tableau $S'$ contains no marked cell on row $j$ (resp. on column $k$). Since $S'\stararrow 
\mathrm{Sat}
(S')$, $
\mathrm{Sat}
(S')$ contains no marked cell on row $j$ (resp. column $k$). Obviously, this is not the case for $
\mathrm{Sat}
(S)$ where the cell $(j,k)$ is marked.
 \cqfd
\end{proof}

\begin{lemma}\label{lemma-final}
For any state $(i,S)$ of $D$, 
$(i,S)\in F_D$ if and only if $(i,\mathrm{Sat}(S))\in F_D$.
\end{lemma}
\begin{proof}
We define the set $F_S$ of the tableaux which have a marked final cell, that is a cell  belonging to a final  row (a final row corresponds to a final state of $Q_B$) or to a final column (a final column corresponds to a final state of $Q_C$) but not to both. The statement is equivalent to $S\in F_S$ if and only if $
\mathrm{Sat}
(S)\in F_S$.  Let $S$ and $S'$ be two $n\times p$ tableaux. It suffices to prove that if $S\rightarrow S'$ then $S'\in F_S$ implies $S\in F_S$. One has to discuss on the position $(j,k)$  of the unique cell which is marked in $S'$ and unmarked in $S$.
\begin{enumerate}
\item If this cell is not final, then  all the final cells in $S$ are unchanged in $S'$. Hence, $S\in F_S$ since $S'\in F_S$.
\item If this cell  is final, on a final row, then  there exists $j'\neq j$ and $k'\neq k$ such that $(j',k)$, $(j,k')$ and $(j',k')$ are the positions of marked cells in $S$ and $S'$. If the cells on $(j',k)$ and $(j,k')$ are not final then the row $j'$ is not final and the column $k'$ is final. Hence the cell on $(j',k')$ is final.
\item If this cell is final on a final column, this case is treated symmetrically to the previous one.
\end{enumerate}
 \cqfd
\end{proof}
  
As a direct consequence  of Lemma~\ref{lem pas separ couples}, we have:
  
  \begin{lemma}\label{lem ajou cell rect lang ok}
  Let $i$ be a state of  $Q_A$ and let $S$ and $S'$ be two tableaux such that $S\rightarrow S'$. Then the two states $(i,S)$ and $(i,S')$ of $D$ are equivalent.
     \end{lemma}
  As a consequence,
     any state $(i,S)$ in $D$ is equivalent to the state $(i,\mathrm{Sat}(S))$. It follows 

    \begin{lemma}\label{lem-3}
     For all $(i,S)\in Q_D$ we have  $(i,\mathrm{Sat}(S))\in [(i,S)]$.
  \end{lemma}
    Let $\alpha_{n,p}$ be the number of saturated tableaux with $n$ lines and $p$ rows. 
 Furthermore, if $i$ is final in $A$, then $S$ contains $(q_0,r_0)$ and consequently so does $\mathrm{Sat}(S)$. 
 Let $\alpha'_{n,p}$ be the number of such tableaux where the cell $(q_0,r_0)$ is marked.  
  Consequently
  
  \begin{theorem}
    $\mathrm{sc}(M\cdot(N\oplus P)) \leq (m-1)\alpha_{n,p}+\alpha'_{n,p}$
  \end{theorem}
  \begin{proof}
By definition, $\mathrm{sc}(M\cdot(N\oplus P))$ is the number of states of the minimal AFD recognizing $M\cdot(N\oplus P)$, that is the number of states of the accessible part of $D\slash \sim$. Trivially, this automaton has less states than  $D\slash \sim$. From Lemma \ref{lem-3}, $D\slash \sim$ has less   than $(m-1)\alpha_{n,p}+\alpha'_{n,p}$ states.
\cqfd
\end{proof}
  
  The next section is devoted to explicit the numbers $\alpha_{n,p}$ and $\alpha'_{n,p}$.
\section{Counting saturated tableaux}\label{se count tab}
According to  previous sections, 
the notion of undistinguishability induces an equivalence relation. Each equivalence class of  tableaux contains a unique saturated tableau.

We first recall some facts about  the classical example of the set partitions (Section \ref{Bell}). The three next sections are devoted to the enumeration of the saturated tableaux. We proceed as follows: We first give an algebraic interpretation of the saturated tableaux in terms of free monoids (Section \ref{monoids}). Hence, in Section \ref{CharSer}, we investigate some properties involving the characteristic series of saturated tableaux and we deduce, in Section \ref{ExpSer}, the desired enumeration. 
\subsection{Bell polynomials\label{Bell}}
Enumeration of saturated tableaux is closely related to set partitions. Well known results about  combinatorial objects and their generating series (Bell polynomials) are recalled here. Readers wishing to deepen their knowledge on set partitions should refer to \cite{Mansour,Mihoubi}. 

Complete multivariate Bell polynomials $A_n(a_1,\ldots,a_n,\ldots)$, named in honor of Eric Temple Bell by Riordan \cite{Rior}, 
are multivariate polynomials defined by their exponential generating series:
\begin{equation}\label{BellComp}
\sum_{n\geq 0}A_n(a_1,\ldots,a_n,\ldots){t^n\over n!}=\exp\left(\sum_{m\geq 1}a_m{t^m\over m!}\right),
\end{equation}
equivalently $A_n(a_1,\ldots,a_n,\ldots)=\left.{d^n\over dt^n}\exp\left(\displaystyle\sum_{m\geq 1}a_m{t^m\over m!}\right)\right|_{t=0}$.

Bell polynomials are themselves multivariate ordinary generating functions of set partitions. Indeed, if $\pi=\{\pi_1,\dots,\pi_k\}$ is a set partition of $\{1,\ldots,n\}$ with $\#\pi_1\geq\cdots\geq\#\pi_k$, we denote by $\lambda_\pi=[\#\pi_1,\dots,\#\pi_k]$ the (integer) partition associated to $\pi$. We have
\[
A_n(a_1,\dots,a_n,\ldots)=\displaystyle\sum_{\pi\vDash n}a_{\#\pi_1}\dots a_{\#\pi_k}=\displaystyle\sum_{\lambda\vdash n}S_\lambda a_{\lambda_1}\cdots a_{\lambda_k}
\]
where 
$S_\lambda$ denotes the number of set partitions $\pi$ such that $\lambda_\pi=\lambda$.
\begin{example}
\rm For instance, we have
$
A_4(a_1,a_2,a_3,a_4)=a_1^4+6a_2a_1^2+3a_2^2+4a_3a_1+a_4,
$
since the partitions of $\{1,\cdots,4\}$ are 

$$\begin{array}{|c|c|}
\hline
\lambda _\pi&\mbox{set partitions}\\
\hline
[1,1,1,1]&\{\{1\},\{2\},\{3\},\{4\}\}\\
\hline
[2,1,1]&\{\{1,2\},\{3\},\{4\}\},\ \{\{1,3\},\{2\},\{4\}\},\\ 
&\{\{1,4\},\{2\},\{4\}\},\{\{2,3\},\{1\},\{4\}\},\\
&\{\{2,4\},\{1\},\{3\}\},\ \{\{3,4\},\{1\},\{2\}\}\\
\hline
[2,2]&\{\{1,2\},\{3,4\}\}, \{\{1,4\},\{2,3\}\},\\
& \{\{1,3\},\{2,4\}\}\\
\hline
[3,1]&\{\{1,2,3\},\{4\}\}, \{\{1,2,4\},\{3\}\},\\
& \{\{1,3,4\},\{2\}\}, \{\{2,3,4\},\{1\}\}\\
\hline
[4]&\{\{1,2,3,4\}\}\\
\hline
\end{array}$$

\end{example}

When each $a_i$ is an integer, one interprets it as the number of ways one can color a part of cardinal $i$ (or equivalently, each $a_i$ represents the weight of each part of cardinal $i$). If each $a_i$ equals $1$ then the complete Bell polynomial $A_n(a_1,\dots,a_n,\ldots)$ specializes to the Bell number $B_n$ which counts the number of set partitions of $\{1,\ldots,n\}$ \emph{i.e.} 
$A_n(1,\ldots,1,\ldots)=B_n.$
Hence, the (exponential) generating series of Bell numbers is
$\sum_{n\geq 0}B_n{t^n\over n!}=\exp(\exp(t)-1)$
 and the first Bell numbers are
$
1, 1, 2, 5, 15, 52, 203, 877, 4140, 21147, 115975,\dots$
(see sequence A000110 in \cite{Sloane}).
From the previous remark, the number of set partitions of $\{1,\ldots,n\}$ with no part of size $1$ equals
$A_n(0,1,\ldots,1,\ldots)$ and the associated generating function is
\begin{equation}\label{SerSansPart1}
\sum_{n\geq 0}A_n(0,1,\ldots,1,\ldots){t^n\over n!}=\exp(\exp(t)-1-t).
\end{equation}

The first values of $A_n(0,1,\dots,1)$ are $1, 0, 1, 1, 4, 11, 41, 162, 715, 3425, 17722, 98253,\dots$
(see sequence A000296 in \cite{Sloane}).
Noticing that 
$
\exp(t)\sum_{n\geq 0}A_n(0,1,\ldots,1,\ldots){t^n\over n!}=\sum_{n\geq 0}B_n{t^n\over n!},
$
we deduce (see  (\ref{binomialtransform})) that the Bell numbers sequence is the binomial transform of the sequence $(A_n(0,1,\dots,1))_{n\in \mathbb{N}}$ \emph{i.e.}
\begin{equation}\label{A2B}
B_n=\sum_{i=0}^n\binom ni A_i(0,1,\ldots,1,\ldots).
\end{equation}
%

When each part have the same weight $x$, the exponential generating function of $B_n(x)=A_n(x,\dots,x)$ is
\begin{equation}\label{Bnalpha}
\sum_{n\geq 0}B_n(x){t^n\over n!}=\exp(x(\exp(t)-1)).
\end{equation}

Notice that each $B_n(x)$ is a degree $n$ polynomial in $x$ also called  Bell polynomial in literature. Its expansion gives
\begin{equation}\label{Bell2Stirling}
B_n(x)=\sum_{k=1}^nS_{n,k}x^k,
\end{equation}
where $S_{n,k}$ is a Stirling number of second kind counting the number of ways of partitioning the set $\{1,\dots,n\}$ into $k$ blocks (see sequence A008277 in \cite{Sloane}). We recover the Bell numbers setting $x=1$, that is $B_n=B_n(1)$. The numbers $r_n=B_n(-1)$ are the Rao Uppuluri-Carpenter numbers \cite{Rao}
whose first terms are $1, -1, 0, 1, 1, -2, -9, -9, 50, 267, 413, -2180, -17731, -50533,\dots$ (see sequence A000587 in \cite{Sloane}). From  (\ref{Bnalpha}) the generating function of the $r_n$'s is
\begin{equation}\label{RaoSer}
\sum_{n\geq 0}r_n{t^n\over n!}=\exp(1-\exp(t)).
\end{equation}

Equality (\ref{Bell2Stirling}) gives $r_n=\sum_{k=1}^n(-1)^kS_{n,k}$  and provides the combinatorial interpretation of the Rao Uppuluri-Carpenter numbers: $r_n$ is the number of set partitions of $\{1,\dots,n\}$ with an even number of parts minus the number of such partitions with an odd number of parts.
\subsection{Monoids of tableaux \label{monoids}}

We consider the set $\mathcal T_n$ of $n$ rows tableaux. It contains a special tableau  $\epsilon_n$ with no column
and we endow it with a  product $\cdot$ such that $T\cdot T'$ is the tableau obtained by gluying $T'$ to the right of $T$. Notice that $T\cdot\epsilon_n=\epsilon_n\cdot T=T$. So, since the product is clearly associative, the set $\mathcal T_n$ is a monoid and more precisely
\begin{lemma}\label{MonCal}
Let $\mathbb C_n=\{c_S\mid S\subset\{1,\dots,n\}\}$ be an alphabet indexed by the subset of $\{1,\dots,n\}$.
The monoid $\mathcal T_n$ is isomorphic to $\mathbb C_n^*$.
\end{lemma}
\begin{proof}
Let $T\in\mathcal T_{n}$ be a tableau with $m$ columns. For any $1\leq j\leq  m$ we define $S_j(T)$ as the set of indices $i$ such that the $i$th cell of the column $j$ is marked. Straightforwardly the map $T\rightarrow c_{S_1(T)}\cdots c_{S_p(T)}$ is an isomorphism of monoid. $\cqfd$
\end{proof}
\begin{example}
Consider the following tableau and its associated word:

\begin{center} 
     \begin{tikzpicture}[scale=0.5]   
	    \foreach \x in {1,...,5} {
	      \foreach \y in {1,...,4} {
	        \pgfmathparse{\x+1} \let\z\pgfmathresult
	        \pgfmathparse{\y+1} \let\t\pgfmathresult
	        \draw[fill=gray!40] (\x,\y) rectangle (\x+1,\y+1);
	      }
	    }  
	    \foreach \y/\x in {4/1,4/4,2/2,2/5,1/1,1/5} {
	        \pgfmathparse{\x+1} \let\z\pgfmathresult
	        \pgfmathparse{\y+1} \let\t\pgfmathresult
	        \draw[fill=white] (\x,\y) rectangle (\x+1,\y+1);
	        \draw (\x,\y) -- (\z,\t);
	        \draw (\z,\y) -- (\x,\t);
	    }
     \end{tikzpicture}
\\    $c_{\{1,4\}}c_{\{3\}}c_{\emptyset}c_{\{1\}}c_{\{3,4\}}$
\end{center}
\end{example}

Consider the subset $\mathbb T_n$ of $\mathbb C_n^*$ consisting of all the words representing saturated tableaux with $n$ rows. 

For each  $S\subset_{\vDash}n$,  we define the mono\"id    $M_S=\left(\{c_\emptyset\}\cup\{c_s|s\in S\}\right)^*$. The following result characterizes the saturated tableaux.

\begin{lemma}
A word belongs to $\mathbb T_n$ if and only if it belongs to a monoid $M_S$ for some $S$.
\end{lemma}
\begin{proof}
Let $w=c_{s_1}\cdots c_{s_p}$ be a word of $\mathbb T_n$. Then $s_i\cup s_j=\emptyset$ or $s_i=s_j$ for each $i\neq j$. Indeed, suppose that $k\in s_i\cap s_j$ and $\ell\in s_i$ such that $\ell\neq k$. Hence, if $T$ denotes the tableau represented by $w$ then  the cells  $(k,i),\ (k,j)$ and $(\ell,i)$ are marked. Since $T$ is a saturated tableau, the cell $(\ell,j)$ is also marked. It follows that $S_w=\{s_i|1\leq i\leq p\}\subset_{\vDash}\{1,\dots,n\}$ and $w\in M_{S_w}$.\\
Conversely, if $w\in M_S$ for some $S\subset_{\vDash}\{1,\dots,n\}$ then $w$ straightforwardly represents a saturated tableau. $\cqfd$
\end{proof}
Equivalently,
\begin{equation}\label{T=cupM}
\mathbb T_n=\bigcup_{S\subset_{\vDash}n}M_S.
\end{equation}
Notice that this union is not disjoint since $M_{S\cap S'}= M_S\cap M_{S'}$.
\subsection{Characteristic series of tableaux\label{CharSer}}
We consider the set $\widetilde M_S=\{w\in M_S|w\not\in M_{S'}\mbox{ for any }S'\subsetneq S\}$.
\begin{example}
For $n=2$ we have
\begin{itemize}
\item[$\bullet$] $\widetilde M_\emptyset=c_\emptyset^*$
\item[$\bullet$] $\widetilde M_{\{\{1\}\}}=M_{\{\{1\}\}}\setminus \widetilde M_\emptyset=M_{\{\{1\}\}}c_{\{1\}}M_{\{\{1\}\}}$
\item[$\bullet$] $\widetilde M_{\{\{2\}\}}=M_{\{\{2\}\}}\setminus \widetilde M_\emptyset=M_{\{\{2\}\}}c_{\{2\}}M_{\{\{2\}\}}$
\item[$\bullet$] $\widetilde M_{\{\{1,2\}\}}=M_{\{\{1,2\}\}}\setminus \widetilde M_\emptyset=M_{\{\{1,2\}\}}c_{\{1,2\}}M_{\{\{1,2\}\}}$
\item[$\bullet$] $\widetilde M_{\{\{1\},\{2\}\}}=M_{\{\{1\},\{2\}\}}\setminus 
\left(\widetilde M_\emptyset \cup 
\widetilde M_{\{\{1\}\}} \cup\widetilde M_{\{\{2\}\}}\right)$
\end{itemize}
\end{example}

So we can rewrite (\ref{T=cupM}) as a disjoint union $\displaystyle{\mathbb T_n=\bigcup_{S\subset_{\vDash}n}\widetilde M_S}$.
The disjoint union implies that the following identity on the characteristic series holds:
\begin{equation}\label{T=sumtildM}
\underline{\mathbb T_n}=\sum_{S\subset_{\vDash}n}\underline{\widetilde M_S}.
\end{equation} 
\begin{example}
We have:
$\underline{\mathbb T_2}=\underline{\widetilde M_\emptyset}+\underline{\widetilde M_{\{\{1\}\}}}
+\underline{\widetilde M_{\{\{2\}\}}}+\underline{\widetilde M_{\{\{1,2\}\}}}+\underline{\widetilde M_{\{\{1\},\{2\}\}}}.$
In the following table we summarize the first terms of each $\underline{\widetilde M_S}$:
\begin{equation}\label{ex4}
\begin{array}{l}
\underline{\widetilde M_\emptyset}=\varepsilon+c_\emptyset+c_\emptyset^2+\dots\\
\underline{\widetilde M_{\{\{1\}\}}}=c_{\{1\}}+c_{\emptyset}c_{\{1\}}+
c_{\{1\}}c_{\emptyset}+c_{\{1\}}^2+\dots\\
\underline{\widetilde M_{\{\{2\}\}}}=c_{\{2\}}+c_{\emptyset}c_{\{2\}}+
c_{\{2\}}c_{\emptyset}+c_{\{2\}}^2+\dots\\
\underline{\widetilde M_{\{\{1,2\}\}}}=c_{\{1,2\}}+c_{\emptyset}c_{\{1,2\}}+
c_{\{1,2\}}c_{\emptyset}+c_{\{1,2\}}^2+\dots\\
\underline{\widetilde M_{\{\{1\},\{2\}\}}}=c_{\{1\}}c_{\{2\}}+
c_{\{2\}}c_{\{1\}}+c_{\{1\}}^2c_{\{2\}}+c_{\{1\}}c_{\{2\}}c_{\{1\}}+c_{\{1\}}c_{\{2\}}^2+c_{\{2\}}c_{\{1\}}^2+\dots
\end{array}
\end{equation}
Hence,
\[\begin{array}{ll}
\underline{\mathbb T_2} = & \varepsilon+c_\emptyset+c_{\{1\}}+c_{\{2\}}+c_{\{1,2\}}+c_\emptyset^2+c_{\{1\}}^2+c_{\emptyset}c_{\{1\}}+c_{\{2\}}^2+c_{\emptyset}c_{\{2\}}+c_{\{2\}}c_{\emptyset}+\\
& c_{\{1,2\}}^2+c_{\emptyset}c_{\{1,2\}}+
c_{\{1,2\}}c_{\emptyset}+c_{\{1\}}c_{\{2\}}+
c_{\{2\}}c_{\{1\}}+\dots.\end{array}
\]
This corresponds to the formal sum of the following saturated tableaux




\noindent
{
 \begin{tikzpicture}[scale=0.3]   
	    \foreach \x in {3,7,11,15} {
	      \foreach \y in {1,2} {
	        \pgfmathparse{\x+1} \let\z\pgfmathresult
	        \pgfmathparse{\y+1} \let\t\pgfmathresult
	        \draw[fill=gray!40] (\x,\y) rectangle (\x+1,\y+1);
	      }
	    }  
	    \foreach \y/\x in {2/7,1/11,1/15,2/15} {
	        \pgfmathparse{\x+1} \let\z\pgfmathresult
	        \pgfmathparse{\y+1} \let\t\pgfmathresult
	        \draw[fill=white] (\x,\y) rectangle (\x+1,\y+1);
	        \draw (\x,\y) -- (\z,\t);
	        \draw (\z,\y) -- (\x,\t);
	    }
	    \node[] at (0,0) {$\varepsilon$};
\node[] at (3.7,0) {$c_\emptyset$};
\node[] at (7.7,0) {$c_{\{1\}}$};
\node[] at (11.7,0) {$c_{\{2\}}$};
\node[] at (15.7,0) {$c_{\{1,2\}}$};
     \end{tikzpicture}
}\ \\
{
 \begin{tikzpicture}[scale=0.3]   
	    \foreach \x in {1,2,5,6,9,10,13,14,17,18,21,22,25,26,29,30,33,34,37,38, 41,42,45,46} {
	      \foreach \y in {1,2} {
	        \pgfmathparse{\x+1} \let\z\pgfmathresult
	        \pgfmathparse{\y+1} \let\t\pgfmathresult
	        \draw[fill=gray!40] (\x,\y) rectangle (\x+1,\y+1);
	      }
	    }  
	    \foreach \y/\x in {2/5,2/10,2/13,2/14,1/17,1/22,1/25,1/26,1/30,2/30,1/33,2/33, 1/37,2/37,1/38,2/38, 2/41,1/42,1/45,2/46} {
	        \pgfmathparse{\x+1} \let\z\pgfmathresult
	        \pgfmathparse{\y+1} \let\t\pgfmathresult
	        \draw[fill=white] (\x,\y) rectangle (\x+1,\y+1);
	        \draw (\x,\y) -- (\z,\t);
	        \draw (\z,\y) -- (\x,\t);
	    }
\node[] at (2,0) {$c_{\emptyset}^2$};
\node[] at (6,0) {$c_{\{1\}}c_{\emptyset}$};
\node[] at (10.2,0) {$c_{\emptyset}c_{\{1\}}$};
\node[] at (14.1,0) {$c_{\{1\}}^2$};
\node[] at (18.2,0) {$c_{\{2\}}c_{\emptyset}$};
\node[] at (22.3,0) {$c_{\emptyset}c_{\{2\}}$};
\node[] at (26.1,0) {$c_{\{2\}}^2$};
\node[] at (30.2,0) {$c_{\emptyset}c_{\{1,2\}}$};
\node[] at (34.2,0) {$c_{\{1,2\}}c_{\emptyset}$};
\node[] at (38.2,0) {$c_{\{1,2\}}^2$};
\node[] at (42.2,0) {$c_{\{1\}}c_{\{2\}}$};
\node[] at (46.3,0) {$c_{\{2\}}c_{\{1\}}$};
     \end{tikzpicture}
}

\end{example}

Notice also that each $w\in M_S$ belongs to one and only one $\widetilde M_{S'}$ with $S'\subset S$. This gives
\begin{equation}\label{M2tildM}
\underline{M_S}=\sum_{S'\subset S}\underline{\widetilde M_{S'}}.
\end{equation}
\begin{example} For $S=\{\{1\},\{2\}\}$ equality (\ref{M2tildM}) is
$\underline{M_{\{\{1\},\{2\}\}}}=\underline{\widetilde M_{\emptyset}}
+\underline{\widetilde M_{\{\{1\}\}}}
+\underline{\widetilde M_{\{\{2\}\}}}
+\underline{\widetilde M_{\{\{1\},\{2\}\}}}
$,
and from (\ref{ex4}), we recover:\\
$\underline{ M_{\{\{1\},\{2\}\}}}=\varepsilon+c_\emptyset+c_{\{1\}}+c_{\{2\}}+
c_\emptyset^2+c_\emptyset c_{\{1\}}+c_\emptyset c_{\{2\}}+
c_{\{1\}}c_\emptyset+c_{\{1\}}^2+c_{\{1\}}c_{\{2\}}+
c_{\{2\}}c_\emptyset+c_{\{2\}}c_{\{1\}}+c_{\{2\}}^2+\cdots$
\end{example}

We define for each $E\subset \{1,\dots,n\}$ the set
\[
\mathbb T_E=\bigcup_{S\subset_{\vDash}E}M_S=\bigcup_{S\subset_{\vDash}E}\widetilde{M_S}.
\]
As for the set $\mathbb T_n$, the last union is disjoint. So we have
\begin{equation}\label{TE=sumtildM}
\underline{\mathbb T_E}=\sum_{S \subset_{\vDash}E}\underline{\widetilde M_S}.
\end{equation} 
We have the following equality for the characteristic series:
\begin{proposition}
\begin{equation}\label{M2TE}
\sum_{S\vDash n}\underline{M_S}=\sum_{E\subseteq \{1,\dots,n\}}A_{n-\#E}(0,1,\ldots,1,\ldots)
\underline{\mathbb T_E}
\end{equation}
\end{proposition}
\begin{proof}
We define
$V_n=\sum_{S\vDash n}\underline{M_S}$,
$\mathrm{support}(S)=\bigcup_{E\in S}E$ for any set of sets $S$ and $\#\#S=\#\mathrm{support}(S)$.
When expanding $V_n$ on the $\underline{\widetilde M_S}$, we obtain:
\[
V_n=\displaystyle\sum_{S\vDash n}\sum_{S'\subset S}\underline{\widetilde M_{S'}}
=\displaystyle\sum_{S'\subset_{\vDash}n}\#\{S\vDash n \mid S'\subset S\}\underline{\widetilde M_{S'}}
\]

Observing that $\#\{S \vDash n\mid S'\subset S\}$ equals the number of set partitions of $\{1,\dots,n\}\setminus\mathrm{support}(S')$ we obtain
\begin{equation}\label{V2tildM}
V_n=\sum_{S\subset_{\vDash}n}B_{n-\#\#S}\underline{\widetilde M_{S}}.
\end{equation}

On the other hand, we define $(a_n)_{n\in \mathbb{N}}$ such that 
$
V_n=\sum_{E\subset \{1,\dots,n\}}a_{n-\#E}\underline{\mathbb{T}_E}.
$

Expanding this equality on the $\underline{\widetilde M_S}$, we find
\[
V_n=\displaystyle\sum_{E\subset \{1,\dots,n\}}a_{n-\#E}\sum_{S\subset_{\vDash}E}\underline{\widetilde M_S}
=\displaystyle\sum_{S\subset_{\vDash}n}\left(\displaystyle\sum_{\mathrm{support}(S)\subset E\subset \{1,\dots,n\}}a_{n-\#E}\right)\underline{\widetilde M_S}
\]

Observing that 
$
\displaystyle\sum_{\mathrm{support}(S)\subset E\subset \{1,\dots,n\}}a_{n-\#E}=\displaystyle\sum_{i=0}^{n-\#\#S}
\binom{n-\#\#S}{i}a_{n-\#\#S-i}\
$, 
we obtain
\begin{equation}\label{V2T}
V_n=\sum_{S\subset_{\vDash}n}\left(\sum_{i=0}^{n-\#\#S}
\binom{n-\#\#S}{n-i}a_{n-\#\#S-i}\right)\underline{\widetilde M_S}.
\end{equation}

Comparing the coefficient of $\underline{\widetilde M_S}$ in (\ref{V2tildM}) and (\ref{V2T}), we find
$
B_k=\sum_{i=0}^k\binom kia_i.
$
That is:  Bell numbers are the binomial transform of the $a_i$'s. From (\ref{A2B}) we deduce $a_i=A_i(0,1,\dots,1,\ldots)$ and our statement.
$\cqfd$\end{proof}
\begin{example}
For $n=2$,
on the left hand we have
\[
\underline{M_{\{\{1\},\{2\}\}}}+\underline{M_{\{\{1,2\}\}}}=
2\underline{\widetilde M_\emptyset}+\underline{\widetilde M_{\{\{1\}\}}}
+\underline{\widetilde M_{\{\{2\}\}}}+\underline{\widetilde M_{\{\{1\},\{2\}\}}}
+\underline{\widetilde M_{\{\{1,2\}\}}}\]
and on the right hand side we have
\[\begin{array}{l}
A_2(0,1,\dots,1,\ldots)\underline{T_\emptyset}+A_1(0,1,\dots,1,\ldots)\left(\underline{T_{\{1\}}}+\underline{T_{\{2\}}}\right)+A_0(0,1,\dots,1,\ldots)\underline{T_{\{1,2\}}}=\\
\underline{T_\emptyset}+\underline{T_{\{1,2\}}}=
2\underline{\widetilde M_\emptyset}+\underline{\widetilde M_{\{\{1\}\}}}
+\underline{\widetilde M_{\{\{2\}\}}}+\underline{\widetilde M_{\{\{1\},\{2\}\}}}
+\underline{\widetilde M_{\{\{1,2\}\}}}.
\end{array}
\]
\end{example}
\subsection{Generating series of saturated tableaux\label{ExpSer}}
Let us first enumerate the elements of $M_S$ with respect to the number of columns (length of the word) and the number of marked cells in the asociated tableaux. Noting that the characteristic series of $M_S$ is given by
\begin{equation}\label{charserM_S}
\underline{M_S}=\sum_{w\in M_S}w={1\over 1-(c_\emptyset+\sum_{E\in S}c_E)},
\end{equation}
we obtain the bivariate generating series sending each letter $c_E$ to the monomial $xt^{\#E}$
\[
m_S(x,t)=\sum_{i,j}\kappa_{i,j}^Sx^it^j,
\]
where $\kappa_{i,j}^S$ denotes the number of words $w=c_{E_1}\dots c_{E_i}$ in $M_S$ of length $i$ and weight $\#E_1+\cdots+\#E_i=j$. Since the map sending each letter $c_E$ to $xt^{\#E}$ is a morphism from $\mathbb Q\langle\langle \mathbb C_S\rangle\rangle$ to 
$\mathbb Q[[x,t]]$, the right hand side of (\ref{charserM_S}) gives
\begin{equation}\label{defgenserM_S}
m_S(x,t)=\frac1{1-(1+P_S(t))x}
\end{equation}
where $P_S(t)=\sum_{E\in S}t^{\#E}$.
\begin{example}
\rm Let $S=\{\{1,4\},\{2\},\{5\}\}$. We have
\[\begin{array}{rcl}
\underline{M_{\{\{1,4\},\{2\},\{5\}\}}}&=&\frac1{1-(c_\emptyset+c_{\{1,4\}}+c_{\{2\}}+c_{\{5\}})}\\
&=&\epsilon+(c_\emptyset+c_{\{1,4\}}+c_{\{2\}}+c_{\{5\}})+(c_\emptyset+c_{\{1,4\}}+c_{\{2\}}+c_{\{5\}})^2+\cdots\\
&=&\epsilon+c_\emptyset+c_{\{1,4\}}+c_{\{2\}}+c_{\{5\}}
+c_\emptyset^2+c_\emptyset c_{\{1,4\}}\\
&&+c_\emptyset c_{\{2\}}+c_\emptyset c_{\{5\}}+
c_{\{1,4\}}c_\emptyset+c_{\{1,4\}}^2+c_{\{1,4\}}c_{\{2\}}+c_{\{1,4\}}c_{\{5\}}\\
&&+c_{\{2\}}c_\emptyset+c_{\{2\}}c_{\{1,4\}}+c_{\{2\}}^2+c_{\{2\}}c_{\{5\}}\\
&&+c_{\{5\}}c_\emptyset+c_{\{5\}}c_{\{1,4\}}+c_{\{5\}}c_{\{2\}}+c_{\{5\}}^2+\cdots
\end{array}\]

Furthermore $P_S(t)=2t+t^2$, so we have
\[
m_{\{\{1,4\},\{2\},\{5\}\}}(x,t)=\frac1{1-(1+2t+t^2)x}=1+(1+2t+t^2)x+(1+4t+6t^2+4t^3+t^4)x^2+\cdots
\]
For instance, we deduce $\kappa_{2,2}^{\{\{1,4\},\{2\},\{5\}\}}=6$ corresponding to the $6$ words: $c_\emptyset c_{\{1,4\}}$, $c_{\{1,4\}}c_\emptyset$, $c_{\{2\}}^2$, $c_{\{2\}}c_{\{5\}}$, $c_{\{5\}}c_{\{2\}}$ and $c_{\{5\}}^2$.
\end{example}

In fact, the series $m_S(x,t)$ depends only on the size of the sets of $S$. Let $\lambda=[\lambda_1,\dots,\lambda_k]$ be a  partition, we define $P_\lambda(t)=\sum_{i=1}^kt^{\lambda_i}$ and $m_\lambda(t)=\frac1{1-(1+P_\lambda(t))x}$. If $S=\{E_1,\dots,E_k\}$ with $\#E_1\geq\cdots\geq \#E_k$, we have
\begin{equation}\label{mS2mlambda}
m_S(x,t)=m_{\lambda_S}(x,t)
\end{equation}
with $\lambda_S=[\#E_1,\dots,\#E_k]$.\\
Denote by $\mathrm{Tab}_{n}(x,t)$ the generating function of the saturated tableaux with $n$ rows
\[
\mathrm{Tab}_n(x,t)=\sum_{m,j}\alpha_{j; n,p}x^pt^j,
\]
where $\alpha_{j;n,p}$ denotes the number of saturated tableaux with $n$ rows, $p$ columns and  $j$ marked cells. So, the coefficient of $x^p$ in $\mathrm{Tab}_n(x,t)$ is a polynomial $\alpha_{n,p}(t)$ which is the generating series of the $n\times p$ saturated tableaux counted by the number of marked cells.

The series $\mathrm{Tab}_{n}(x,t)$ is the image of $\underline{\mathbb T_n}$ by the morphism $c_E\rightarrow xt^{\#E}$.  Notice that the image of $\underline{\mathbb T_E}$ is also $\mathrm{Tab}_{\#E}(x,t)$.

Hence Equation (\ref{M2TE}) becomes
\begin{equation}\label{m2tab1}
\sum_{S \vDash n}m_{S}(x,t)=\sum_{E\subset\{1,\dots,n\}}A_{n-\#E}(0,1,\dots,1)\mathrm{Tab}_{\#E}(x,t).
\end{equation}

The number of set partitions $S$ of $\{1,\dots,n\}$ such that $\lambda_S=\lambda$ equals $n!\over \lambda!$ with $\lambda!=(\prod_i(\lambda_i!))(\prod_i(\mathrm{mult}_i(\lambda)!))$ where $\mathrm{mult}_i(\lambda)$ denotes the multiplicities of the part $i$ in $\lambda$.
\begin{example}\rm
Consider $\lambda=[2,1,1,1]$. The number of set partitions of $S$ such that $\lambda_S=\lambda$ equals ${5!\over (2!1!^3)(1!3!)}=10$. The corresponding partitions are:
$$\begin{array}{cc}
\{\{1,2\},\{3\},\{4\},\{5\}\}\hspace*{2cm}&\{\{1,3\},\{2\},\{4\},\{5\}\}\\
\{\{1,4\},\{2\},\{3\},\{5\}\}\hspace*{2cm}&\{\{1,5\},\{2\},\{3\},\{4\}\}\\
\{\{2,3\},\{1\},\{4\},\{5\}\}\hspace*{2cm}&\{\{2,4\},\{1\},\{3\},\{5\}\}\\
\{\{2,5\},\{1\},\{3\},\{4\}\}\hspace*{2cm}&\{\{3,4\},\{1\},\{2\},\{5\}\}\\
\{\{3,5\},\{1\},\{2\},\{4\}\}\hspace*{2cm}&\{\{4,5\},\{1\},\{2\},\{3\}\}.
\end{array}$$
\end{example}

Furthermore, the number of subset of $\{1,\dots,n\}$ with a fixed size $i$ is equal to $\binom ni$. Hence from (\ref{m2tab1}) we obtain
\begin{equation}\label{m2tab2}
\sum_{\lambda\vdash n}{n!\over\lambda!}m_\lambda(x,t)=\sum_{i=0}^n\binom niA_{i}(0,1,\dots,1)\mathrm{Tab}_{n-i}(x,t).
\end{equation}
On the right hand of this equality, we recognize the binomial convolution of the sequence $(A_{i}(0,1,\dots,1))_i$ with the sequence $({\mathrm{Tab}}_{n-i}(x,t))_i$ (see Equality (\ref{binconv})). Interpreting (\ref{m2tab2}) in terms of exponential generating functions we set
\[
m(x,t,z)=\sum_{n=0}^\infty\left(\sum_{\lambda\vdash n}{1\over\lambda!}m_\lambda(x,t)\right)z^n
\] 
and
\[
\mathrm{Tab}(x,t,z)=\sum_{n=0}^\infty\mathrm{Tab}_n(x,t){z^n\over n!}
\]
Hence, using   (\ref{SerSansPart1}), Equality (\ref{m2tab2}) becomes
\begin{equation}\label{m2tab3}
m(x,t,z)=\exp(\exp(z)-z-1)\mathrm{Tab}(x,t,z).
\end{equation}
So,
\begin{equation}\label{tab2m}
\mathrm{Tab}(x,t,z)=m(x,t,z)\exp(-\exp(z)+z+1)=-m(x,t,z){d\over dz}\exp(1-\exp(z)).
\end{equation}
We recognize the exponential generating function of the Rao Uppuluri-Carpenter numbers (\ref{RaoSer}). As the coefficient of $x^pz^n$ in the left hand-side and the right hand-side of (\ref{tab2m}) are equal, we obtain
\begin{theorem}\label{alphatrao}
The generating function of the $n\times m$ saturated tableaux counted by numbers of marked cells is
\begin{equation}\label{alphatraoeq}
\alpha_{n,p}(t)=-n!\sum_{i=0}^n{r_{n-i+1}\over (n-i)!}\sum_{\lambda\vdash i}{(1+P_\lambda(t))^p\over \lambda!}
\end{equation}
\end{theorem}
\begin{example}
\rm Let the first values of the polynomial $\alpha_{n,p}(t)$:
\begin{itemize}
\item[$\bullet$]$\alpha_{1,p}(t)=(1+t)^p$,
\item[$\bullet$]$\alpha_{2,p}(t)=(1+2t)^p+(1+t^2)^p-1$, 
\item[$\bullet$]$\alpha_{3,p}(t)=(3t+1)^p+3(t+t^2+1)^p-3(1+t)^p+(t^3+
1)^p-1$, 
\item[$\bullet$]$\alpha_{4,p}(t)=2+(1+4t)^p+6(1+2t+t^2)^p+3(1+2t^2)^p+4(1+t+t^3)^p+(1+t^4)^p-4(1
+t)^p-6(1+2t)^p-6(1+t^2)^p$,
\item[$\bullet$]$\alpha_{5,p}(t)= 9+10(1+t)^p+(5t+1)^p+10(3t+t^2+1)^p-10(3t
+1)^p+15(t+2t^2+1)^p-30(t+t^2+1)^p+10(2t+t^3+1)^p-10(1+2t)^p+10(t^2+t^
3+1)^p-10(1+t^2)^p-10(t^3+1)^p+5(t+t^4+1)^p+(t^5+1)^p$
\end{itemize}
For instance, the coefficient of $t^9$ in $\alpha_{4,3}$ is $4$. This corresponds to the tableaux
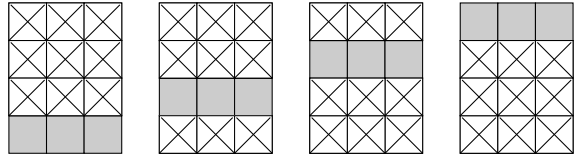
\begin{center}
 \begin{tikzpicture}[scale=0.5]   
	    \foreach \x in {1,2,3,5,6,7,9,10,11,13,14,15} {
	      \foreach \y in {1,2,3,4} {
	        \pgfmathparse{\x+1} \let\z\pgfmathresult
	        \pgfmathparse{\y+1} \let\t\pgfmathresult
	        \draw[fill=gray!40] (\x,\y) rectangle (\x+1,\y+1);
	      }
	    }  
	    \foreach \y/\x in {4/1,4/2,4/3,3/1,3/2,3/3,2/1,2/2,2/3,
4/5,4/6,4/7,3/5,3/6,3/7,1/5,1/6,1/7,
4/9,4/11,4/10,1/9,1/11,1/10,2/9,2/11,2/10,
1/13,1/14,1/15,3/13,3/14,3/15,2/13,2/14,2/15} {
	        \pgfmathparse{\x+1} \let\z\pgfmathresult
	        \pgfmathparse{\y+1} \let\t\pgfmathresult
	        \draw[fill=white] (\x,\y) rectangle (\x+1,\y+1);
	        \draw (\x,\y) -- (\z,\t);
	        \draw (\z,\y) -- (\x,\t);
	    }
     \end{tikzpicture}
\end{center}

We illustrate the fact that  $\alpha_{n,p}(t)=\alpha_{p,n}(t)$ by remarking
\[\begin{array}{l}
\alpha_{3,4}(t)=(3t+1)^4+3(t+t^2+1)^4-3(1+t)^4+(t^3+
1)^4-1=\\
\alpha_{4,3}(t)=2+(1+4t)^3+6(1+2t+t^2)^3+3(1+2t^2)^3\\+4(1+t+t^3)^3+(1+t^4)^3-4(1
+t)^3-6(1+2t)^3-6(1+t^2)^3=\\
{t}^{12}+4\,{t}^{9}+3\,{t}^{8}+12\,{t}^{7}+36\,{t}^{6}+48\,{t}^{5}+135
\,{t}^{4}+148\,{t}^{3}+66\,{t}^{2}+12\,t+1\end{array}
\]
\end{example}
Setting $t=1$  in $\alpha_{n,m}(t)$ we obtain the number of $n\times m$ saturated tableaux.
\begin{corollary}
The number of $n\times m$ saturated tableaux is
\begin{equation}\label{alphanm}
\alpha_{n,m}=\alpha_{n,m}(1)=-n!\sum_{k=0}^n\sum_{\omega(\lambda)<=n\atop \#\lambda=
k}{r_{n-\omega(\lambda)+1}\over \lambda!(n-\omega(\lambda))!}(k+1)^m,
\end{equation}
where $\omega(\lambda)=k$ if $\lambda\vdash k$.
\end{corollary}

\begin{example}
Let the first value of $\alpha_{n,m}$ below:
\[\begin{array}{l@{\ \ }l@{\ \ }l@{\ \ }l@{\ \ }l@{\ \ }l@{\ \ }l@{\ \ }l}
1\\
1&1\\
1&4&12\\
1&8&34&128\\
1&16&96&466&2100\\
1&32&274&1688& 9226& 48032\\
1& 64& 792& 6154& 40356& 245554& 1444212\\
1& 128& 2314& 22688& 177466& 1251128& 8380114& 54763088
\end{array}
\]
\end{example}
\begin{proposition}
The  number of tableaux having a specific marked cell in a $n\times m$ saturated tableau equals
\[
\alpha'_{n,p}=\frac1{np}\left.{d\over dt}\alpha_{n,p}(t)\right|_{t=1}
=-{(n-1)!}\sum_{i=0}^n{r_{n-i+1}\over (n-i)!}\sum_{\lambda\vdash i}{P'_\lambda(1)(1+P_\lambda(1))^{p-1}\over \lambda!}
\]
\end{proposition}
\begin{proof}
Let $\alpha_{n,p}^{(i,j)}$ be the number of $n\times p$ triangle free tableaux such that the cell $(i,j)$ is marked. Each cell having the same role, we have $\alpha_{n,p}^{(i,j)}=\alpha_{n,p}^{(0,0)}=\alpha'_{n,p}$. Hence $$\alpha'_{n,p}=\frac{1}{np}\displaystyle\sum_{i=0}^{n-1}\sum_{j=0}^{p-1}\alpha_{n,p}^{(i,j)}=\frac{1}{np}\sum_T\mathcal{M}(T)$$ where the last sum runs over the $n\times p$ triangle free tableaux $T$ and $\mathcal{M}(T)$ denotes the number of marked cells in $T$. Notice that $\sum_T\mathcal{M}(T)$ is the total number of marked cells in all the $n\times p$ triangle free tableaux.
Acting by $t{d\over dt}$ multiplies the coefficient of $t^j$ in $\alpha_{n,m}(t)$ by $j$ so that
$\left.t{d\over dt}\alpha_{n,p}(t)\right|_{t=1}=\left.{d\over dt}\alpha_{n,p}(t)\right|_{t=1}$ is the number of marked cells in all the $n\times p$ tableaux. We deduce $\alpha'_{n,p}=\frac1{np}\left.{d\over dt}\alpha_{n,p}(t)\right|_{t=1}$.\cqfd

\end{proof}
\begin{example}
We list below the values of $\alpha'_{n,p}$ for $n=1, \dots, 6$:
\begin{itemize}
\item[$\bullet$]$\alpha'_{1,p}={2}^{p-1}$,
\item[$\bullet$]$\alpha'_{2,p}={3}^{p-1}+{2}^{p-1}$,
\item[$\bullet$]$\alpha'_{3,p}={4}^{p-1}+3\cdot{3}^{p-1}$,
\item[$\bullet$]$\alpha'_{4,p}={5}^{p-1}+6\cdot{4}^{p-1}+4\cdot{3}^{p-1}-3\cdot{2}^{p-1}$,
\item[$\bullet$]$\alpha'_{5,p}=-7\cdot{2}^{p-1}+{6}^{p-1}+10\cdot{5}^{p-1}+19\cdot{4}^{p-1}-7\cdot{3}^{p-1}$,
\item[$\bullet$]$\alpha'_{6,p}=-59\cdot{3}^{-1+p}+20\cdot{4}^{p-1}+55\cdot{5}^{p-1}+{7}^{p-1}+ 15\cdot{6}^{p-1}$.
\end{itemize}
For instance, consider the saturated $2\times 2$ tableaux:
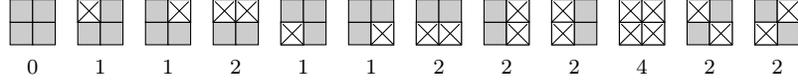
\begin{center}
\begin{tikzpicture}[scale=0.3]   
	    \foreach \x in {1,2,4,5,7,8,10,11,13,14,16,17,19,20,22,23,25,26,28,29,31,32,34,35} {
	      \foreach \y in {1,2} {
	        \pgfmathparse{\x+1} \let\z\pgfmathresult
	        \pgfmathparse{\y+1} \let\t\pgfmathresult
	        \draw[fill=gray!40] (\x,\y) rectangle (\x+1,\y+1);
	      }
	    }  
	    \foreach \y/\x in {2/4,2/8,2/10,2/11,1/13,1/17,1/19,1/20,1/23,2/23,2/25,1/25,1/28,2/28,1/29,2/29,2/31,1/32,1/34,2/35} {
	        \pgfmathparse{\x+1} \let\z\pgfmathresult
	        \pgfmathparse{\y+1} \let\t\pgfmathresult
	        \draw[fill=white] (\x,\y) rectangle (\x+1,\y+1);
	        \draw (\x,\y) -- (\z,\t);
	        \draw (\z,\y) -- (\x,\t);
	    }
\node[] at (2,0) {0};\node[] at (29,0) {4};
\foreach \x in {5,8,14,17}{\node[] at (\x,0) {1};}
\foreach \x in {11,20,23,26,32,35}{\node[] at (\x,0) {2};}
     \end{tikzpicture}
\end{center}
We find $\alpha_{2,2}'=\frac14(1\cdot 0+4\cdot 1+6\cdot 2+1\cdot 4)=5=3^1+2^1$ as expected by the formula.
\end{example}

\section{Witness}\label{sec wit}
In this section, we use the Brzozowski automata $A=W_m(a,c,b,d)$, $B=W_n(a,b,c,d)$ and $C=W_p(d,b,c,a)$ (with $m,n,p\geq 3$) mentioned in Section~\ref{background}, page~\pageref{Brzo-def} and recognizing three languages $M$, $N$ and $P$. They are represented in Figure \ref{AFD_A} where the red arrows suggest how to connect the automaton $A$ to the automata $B$ and $C$ to recognize $M\cdot(N\oplus P)$. 
 \begin{figure}[H]
	\centerline{
		\begin{tikzpicture}[node distance=1.5cm, bend angle=25]
			\node[state,initial] (0)  {$0$};
			\node[state] (1) [right of=0] {$1$};
			\node[state] (2) [right of=1] {$2$};
			\node (etc1) [right of=2] {$\ldots$};
			\node[state, rounded rectangle] (m-3) [right of=etc1] {$m-3$};
			\node[state, rounded rectangle] (m-2) [right of=m-3] {$m-2$};
			\node[state, rounded rectangle, accepting] (m-1) [right of=m-2] {$m-1$};
			\node[state,rectangle,red, text height= 8pt](ac)[right of=m-1]{$a,c$}; 
			\node[state,rectangle,red, text height= 8pt](bd)[right of=ac]{$b,d$}; 
			\path[->]
        (0) edge[bend left] node {$a$} (1)
        (1) edge[bend left] node {$a$} (2)
        (2) edge[bend left] node {$a$} (etc1)
        (etc1) edge[bend left] node {$a$} (m-3)
        (m-3) edge[bend left] node {$a$} (m-2)
        (m-2) edge[bend left] node {$a, c$} (m-1)
                  edge[very thick,dashed,-, red,out=270,in=270,looseness=.4] node {}(ac)
        (m-1) edge[out=-115, in=-65, looseness=.2] node[above] {$a$} (0)
             edge[very thick,dashed, -, red,out=45,in=145,looseness=.4] node {}(bd)
		    (0) edge[out=115,in=65,loop] node {$b, c, d$} (0)
		    (1) edge[out=115,in=65,loop] node {$c, d$} (1)
		    (2) edge[out=115,in=65,loop] node {$b, c, d$} (2)
		    (m-3) edge[out=115,in=65,loop] node {$b, c, d$} (m-3)
		    (m-2) edge[out=115,in=65,loop] node {$b, d$} (m-2)
		    (m-1) edge[out=115,in=65,loop] node {$b, d$} (m-1)
        (1) edge[bend left] node[above] {$b$} (0)
        (m-1) edge[bend left] node[above] {$c$} (m-2)
			;
			\node[state] (q0) [above of=etc1, node distance=2cm] {$q_0$};
			\node[state] (q1) [right of=q0] {$q_1$};
			\node[state] (q2) [right of=q1] {$q_2$};
			\node (etc2) [right of=q2] {$\ldots$};
			\node[state, rounded rectangle] (n-3) [right of=etc2] {$q_{n-3}$};
			\node[state, rounded rectangle] (n-2) [right of=n-3] {$q_{n-2}$};
			\node[state, rounded rectangle, accepting] (n-1) [right of=n-2] {$q_{n-1}$};
			\path[->]
        (q0) edge[bend left] node {$a$} (q1)
        (q1) edge[bend left] node {$a$} (q2)
        (q2) edge[bend left] node {$a$} (etc2)
        (etc2) edge[bend left] node {$a$} (n-3)
        (n-3) edge[bend left] node {$a$} (n-2)
        (n-2) edge[bend left] node {$a, b$} (n-1)
        (n-1) edge[out=-115, in=-65, looseness=.2] node[above] {$a$} (q0)
		    (q0) edge[out=115,in=65,loop] node {$b, c, d$} (q0)
		    (q1) edge[out=115,in=65,loop] node {$b, d$} (q1)
		    (q2) edge[out=115,in=65,loop] node {$b, c, d$} (q2)
		    (n-3) edge[out=115,in=65,loop] node {$b, c, d$} (n-3)
		    (n-2) edge[out=115,in=65,loop] node {$c, d$} (n-2)
		    (n-1) edge[out=115,in=65,loop] node {$c, d$} (n-1)
        (q1) edge[bend left] node[above] {$c$} (q0)
        (n-1) edge[bend left] node[above] {$b$} (n-2)
			;
			\node[state] (r0) [below of=etc1, node distance=2cm] {$r_0$};
			\node[state] (r1) [right of=r0] {$r_1$};
			\node[state] (r2) [right of=r1] {$r_2$};
			\node (etc3) [right of=r2] {$\ldots$};
			\node[state, rounded rectangle] (p-3) [right of=etc3] {$r_{p-3}$};
			\node[state, rounded rectangle] (p-2) [right of=p-3] {$r_{p-2}$};
			\node[state, rounded rectangle, accepting] (p-1) [right of=p-2] {$r_{p-1}$};
			\path[->]
        (r0) edge[bend left] node {$d$} (r1)
        (r1) edge[bend left] node {$d$} (r2)
        (r2) edge[bend left] node {$d$} (etc3)
        (etc3) edge[bend left] node {$d$} (p-3)
        (p-3) edge[bend left] node {$d$} (p-2)
        (p-2) edge[bend left] node {$d, b$} (p-1)
        (p-1) edge[out=-115, in=-65, looseness=.2] node[above] {$d$} (r0)
		    (r0) edge[out=115,in=65,loop] node {$a, b, c$} (r0)
		    (r1) edge[out=115,in=65,loop] node {$a, b$} (r1)
		    (r2) edge[out=115,in=65,loop] node {$a, b, c$} (r2)
		    (p-3) edge[out=115,in=65,loop] node {$a, b, c$} (p-3)
		    (p-2) edge[out=115,in=65,loop] node {$a, c$} (p-2)
		    (p-1) edge[out=115,in=65,loop] node {$a, c$} (p-1)
        (r1) edge[bend left] node[above] {$c$} (r0)
        (p-1) edge[bend left] node[above] {$b$} (p-2)
			;
			\node (inv) [above of=m-3, node distance=1.1cm] {};
			\draw[-, red] (ac.45) .. controls +(-.1,1) and +(-.1,0) .. (inv.center);
			\draw[->, red] (inv.center) .. controls +(-1,0) and +(-.3,-.8) .. (q0);
			;
			\node (inv2) [above of=m-2, node distance=1.2cm] {};
			\draw[-, red] (bd.45) .. controls +(-.1,.9) and +(-.1,0) .. (inv2.center);
			\draw[->, red] (inv2.center) .. controls +(-1,0.1) and +(0,-.8) .. (q0.south);
			;
			\node (inv3) [below of=etc1, node distance=.8cm] {};
			\draw[-, red] (ac.-45) .. controls +(-.1,-.7) and +(-.1,0) .. (inv3.center);
			\draw[->, red] (inv3.center) .. controls +(-.3,0) and +(-.7,.8) .. (r0.160);
			;
			\node (inv4) [below of=m-2, node distance=.9cm] {};
			\draw[-, red] (bd.-45) .. controls +(0,-.7) and +(-.1,0) .. (inv4.center);
			\draw[->, red] (inv4.center) .. controls +(-2,0) and +(0.2,.8) .. (r0.20);
    \end{tikzpicture}
  }
  \caption{The automata $A, B$ and $C$}
  \label{AFD_A}
\end{figure}
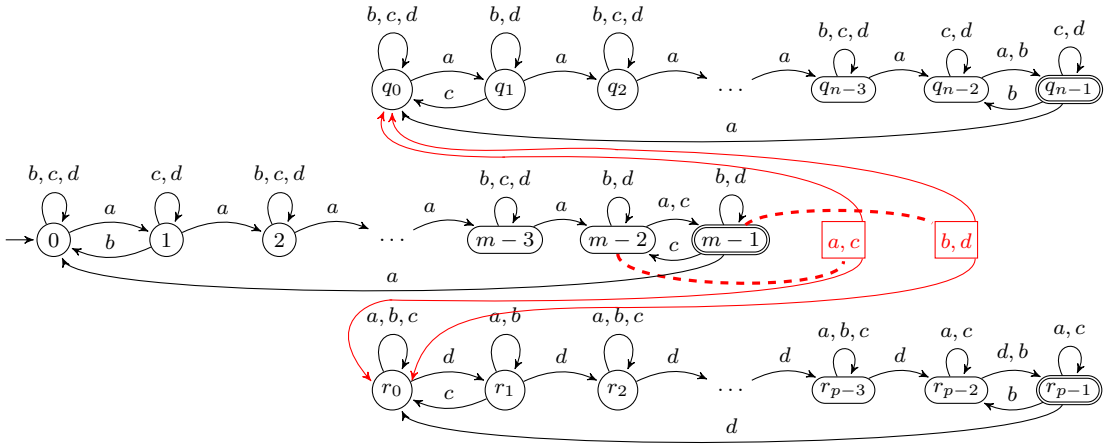

  It is obvious that $A$, $B$ and $C$ are minimal since they contain only one final state and, in any of them, one symbol acts as a cycle over the set of states. To prove $(M,N,P)$ to be a witness for the composition of catenation with symmetric difference, we have to prove that the DFA $D$, built from $A, B$ and $C$ by following the construction described in Section \ref{sec prelim} reaches the bound computed in Section~\ref{upper}. Some properties of $A, B$ and $C$ are useful. Especially, one can obtain rotations and groupings on their states reading some particular words.
	
  
  The following properties are true for $W_n(a,b,c,d)$, but they remain valid for the other automata by just permuting role of the symbols: 

\begin{itemize}
	\item  
	The word $a^k$ induces a $k$-rotation over $Q_B$.

	
	\item  
	The word $(ab)^{(k-2)}$ induces a grouping of $q_{n-k}$ and $q_{n-2}$ sending these two states to $q_{n-1}$, $q_{n-2}$ (with $2< k\leq n$). The word $(ba)^{(k-1)}$  induces another grouping which sends the states $q_{n-k}$, $q_{n-1}$ to the states $q_0$, $q_{n-1}$ (with $2\leq k\leq n-1$).
	
\end{itemize}
\begin{remark}\label{Claim::permute}
Any word in $\{a,b,d\}^*$ acts as a permutation on $Q_B\times Q_C$. Consequently for any word $w\in\{a,b,d\}^*$ and any set $S\subset Q_B\times Q_C$, $|w.S|=|S|=|S.w|$. 
\end{remark}
\begin{proposition}\label{prop D access}
Each state of $D$ is accessible.
\end{proposition}

\begin{proof}
By induction over the structure of the states of $D$. The initial state is $(0,\emptyset)$ and any state of the form $(i,\emptyset)$ with $i<m-1$ is reached from it by $a^i$.

Let $\alpha\in[0,np-1]$ and consider now any state $(i,S)$ with $|S|=\alpha+1$. We show its accessibility by cases, using only words in $\{a,b,d\}^*$ acting as permutations according to Remark~\ref{Claim::permute}.
\begin{enumerate}
	\item\label{it1} If $i=m-1$ then $S=\{(q_0,r_0)\}\uplus S'$ and
		the state $(i,S)$ is reached by $a$ from $(m-2,a\cdot S')$, itself accessible by induction hypothesis.
	\item If $i=0$ then we have three cases to consider: 
	\begin{enumerate}
		\item\label{it2a} If $S=\{(q_1,r_0)\}\uplus S'$ then $(0,S)$ is reached by $a$ from $(m-1,\{(q_0,r_0)\}\uplus a\cdot S')$ which is accessible by \ref{it1}.
		\item\label{it2b} If $S=\{(q_1,r_k)\}\uplus S'$ then $(0,S)$ is reached by $d^k$ from $(0,\{(q_1,r_0)\}\uplus d^k\cdot S')$ which is accessible by \ref{it2a}.
		\item\label{it2c}
		Otherwise, $S=\{(q_j,r_k)\}\uplus S'$. Then the word $w=(abb)^{(j-1)\mbox{\tiny{ mod }}n}$ allows to reach $(0,S)$ from $(0,\{(q_1,r_k)\}\uplus w\cdot S')$ which is accessible by \ref{it2b}.
		\end{enumerate}
	\item For any $0<i<m-1$, the state $(i,S)$ is reached by $a^i$ from $(0,a^i\cdot S)$ which is accessible by \ref{it2c}.
\end{enumerate}
\cqfd
\end{proof}

As a direct consequence,

\begin{theorem}
Each state of $D\slash\sim$ is accessible.
\end{theorem}

So $D\slash\sim$ is the minimal DFA of $M\cdot(N\oplus P)$. Furthermore, we know from 
Lemma~\ref{lem-3}
that any state $[(i,S)]$ of $D\slash\sim$ contains $(i,\mathrm{Sat}(S))$. In order $\mathrm{sc}(M\cdot(N\oplus P))$ to reach $(m-1)\alpha_{n,p}+\alpha'_{n,p}$, it remains to prove the non equivalence of any two distinct saturated states of $D$, where a state $(i,S)$ is \emph{saturated} if $S$ is a saturated tableau.

\begin{theorem}
Any two distinct saturated states of $D$ are non-equivalent.
\end{theorem}

\begin{proof}
Let $s_1=(i_1,S_1)$ and $s_2=(i_2,S_2)$ be two distinct saturated states.
There are two cases to consider:
 If $i_1\neq i_2$, let us suppose $i_1>i_2$ (without loss of generality). Then, using the word $w=ccd^{p-1}a^{m-1-i_1}$, we send $s_1$ to $(i_1',S_1')$ and $s_2$ to $(i_2',S_2')$ with $(q_0,r_0)\in 
 \mathrm{Sat}
 (S_1')\backslash 
 \mathrm{Sat}
 (S_2')$, reducing this case to the other one, next described. For an explanation, the prefix $cc$ acts as a contraction which allows, in particular, to exclude $r_1$ of $S_2$ with no move on $A$. 
 Then the factor $d^{p-1}$ induces a rotation on $C$ to exclude $r_0$ of $S_2$, still with no move on $A$. 
  Last, the suffix $a^{m-1-i_1}$ takes $i_1$ to $m-1$, involving the emergence of the couple $(q_0,r_0)$ in $S_1'$, whereas the state $r_0$ remains excluded of $S_2'$. Then, following Lemma~\ref{lemma-distinct},  $\mathrm{Sat}(S_1')\neq\mathrm{Sat}(S_2')$.

Now, we look at the other case. If $S_1\neq S_2$, let $(q_j,r_k)\in S_1\backslash S_2$ (without loss of generality). A state $(i,S)$ is final if and only if $S$ holds a \textit{final} couple of the form $(q_{n-1},r_k)$ with $k\neq p-1$ or $(q_j,r_{p-1})$ with $j\neq n-1$.  We will prove it is possible to build a word sending $(q_j,r_k)$ on such a final couple while sending  any other couple on a non-final couple. And then we will be able to conclude following Lemma~\ref{lemma-final}.

 First,  we define the sets $R=\{r_{k'}|(q_j,r_{k'})\in S_2\}$ and $T=\{q_j\}\cup\{q_{j'}|(q_{j'},r_{k'})\in S_2\mbox{ for some }r_{k'}\in R\}$  and we show the following lemma 
 \begin{lemma}\label{lemmaRT}
 One has $T\times R\subset S_2$.
 \end{lemma}
 \begin{proof}
 Let $(q,r)\in T\times R$. If $q=q_j$ then the property is obvious. Else $(q_j,r)\in S_2$ and there exists $r'\in R$ such that $(q,r')\in S_2$. If $r=r'$ then again the lemma is straightforward. In the other cases, one has $(q_j,r')\in S_2$ and since $S_2$ is saturated we also have $(q,r)\in S_2$.\cqfd   
 \end{proof}
 
 Now, we partition the states of $B$ and $C$ using a coloration:
  we color in green the states of $R\cup T$, in black the other states appearing in $S_2$ plus the state $r_k$ and in grey the states not used in $S_2$ (see 
  Figure~\ref{avant},
   for an illustration). In fact, the colors spot states on the automata $B$ and $C$ and, by reading a word, we move the colors on other states.

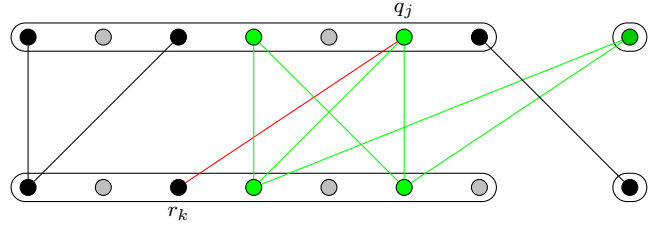
\begin{figure}[H]
  \centerline{
	\begin{tikzpicture}[node distance=1cm,shorten >=0pt]
      \node[state,fill=black] (q0) {};
      \node[state,right of=q0,fill=lightgray] (q1) {};
      \node[state,right of=q1,fill=black] (q2) {};
      \node[state,right of=q2,fill=green] (q3) {};
      \node[state,right of=q3,fill=lightgray] (q4) {};
      \node[state,right of=q4,fill=green] (q5) {};
      \node[above of=q5,node distance=0.35cm] (qj) {$q_j$};
      \node[state,right of=q5,fill=black] (q6) {};
      \node[state,right of=q6,node distance=2cm,fill=black!20!green] (q7) {};
      \node[state,below of=q0,node distance=2cm,fill=black] (r0) {};
      \node[state,right of=r0,fill=lightgray] (r1) {};
      \node[state,right of=r1,fill=black] (r2) {};
      \node[below of=r2,node distance=0.35cm] (rk) {$r_k$};
      \node[state,right of=r2,fill=green] (r3) {};
      \node[state,right of=r3,fill=lightgray] (r4) {};
      \node[state,right of=r4,fill=green] (r5) {};
      \node[state,right of=r5,fill=lightgray] (r6) {};
      \node[state,right of=r6,node distance=2cm,fill=black] (r7) {};
      \node[state,rounded rectangle,fit=(q0) (q6)] (QA) {};
      \node[state,rounded rectangle,fit=(q7)] (FA) {};
      \node[state,rounded rectangle,fit=(r0) (r6)] (QB) {};
      \node[state,rounded rectangle,fit=(r7)] (FB) {};
      \path
        (q5) edge[color=red] (r2)
        (q0) edge (r0)
        (q2) edge (r0)
        (q3) edge[color=green] (r3)
        (q3) edge[color=green] (r5)
        (q5) edge[color=green] (r3)
        (q5) edge[color=green] (r5)
        (q7) edge[color=green] (r3)
        (q7) edge[color=green] (r5)
        (q6) edge (r7)
        ;
    \end{tikzpicture}
  }
  \caption{An example of coloration for states of $B$ and $C$.}
  \label{avant}
\end{figure}
So we associate a pair of colors to each couple in $S_2$ and Lemma \ref{lemmaRT} implies that any couple of $S_2$ is either green-green or black-black. The couple $(q_j,r_k)$, only present in $S_1$, is green-black.

In the remaining of the proof, we compute a word sending
\begin{enumerate}
	\item\label{point1} the couple $(q_j,r_k)$ on a final couple of the form $(q_{n-1},r_{k'})$ with $k'\neq p-1$,
	\item\label{point2} all green-green couples on the non-final $(q_{n-1},r_{p-1})$ and 
	\item\label{point3} all black-black couples on non-final couples of the form $(q_{j'},r_{k'})$ with $j'\neq n-1$ and $k'\neq p-1$.
\end{enumerate}
Clearly, such a word separates $s_1$ from $s_2$, sending the first one to a final state (because of point \ref{point1}) and the second one to a non-final state (because of points \ref{point2} and \ref{point3}). For instance, the colors in 
Figure~\ref{avant}
 become as in 
 Figure~\ref{après}. 
\begin{figure}[H]
  \centerline{
	\begin{tikzpicture}[node distance=1cm,shorten >=0pt]
      \node[state,fill=lightgray] (q0) {};
      \node[state,right of=q0,fill=black] (q1) {};
      \node[state,right of=q1,fill=black] (q2) {};
      \node[state,right of=q2,fill=black] (q3) {};
      \node[state,right of=q3,fill=lightgray] (q4) {};
      \node[state,right of=q4,fill=lightgray] (q5) {};
      \node[state,right of=q5,fill=lightgray] (q6) {};
      \node[state,right of=q6,node distance=2cm,fill=black!20!green] (q7) {};
      \node[state,below of=q0,node distance=2cm,fill=lightgray] (r0) {};
      \node[state,right of=r0,fill=black] (r1) {};
      \node[state,right of=r1,fill=lightgray] (r2) {};
      \node[state,right of=r2,fill=black] (r3) {};
      \node[state,right of=r3,fill=lightgray] (r4) {};
      \node[state,right of=r4,fill=lightgray] (r5) {};
      \node[state,right of=r5,fill=black] (r6) {};
      \node[state,right of=r6,node distance=2cm,fill=green] (r7) {};
      \node[state,rounded rectangle,fit=(q0) (q6)] (QA) {};
      \node[state,rounded rectangle,fit=(q7)] (FA) {};
      \node[state,rounded rectangle,fit=(r0) (r6)] (QB) {};
      \node[state,rounded rectangle,fit=(r7)] (FB) {};
      \path
        (q7) edge[color=red] (r6)
        (q1) edge (r1)
        (q2) edge (r1)
        (q7) edge[color=green] (r7)
        (q3) edge (r3)
        ;
    \end{tikzpicture}
  }
  \caption{A configuration distinguishing the states which correspond to the coloration of Figure~\ref{avant}.}
  \label{après}
\end{figure}
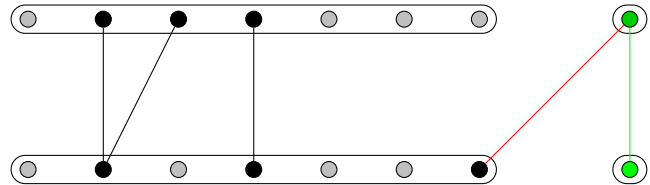

To prove the existence of a word satisfying the three points above, we proceed by induction on the number of green states. For each step we will use the notions of rotation and grouping described
in Section~\ref{background}
in order to decrease the number of green states. The difficulty lies not only in decreasing this number, but also in avoiding to increase it inadvertently. Indeed, we have to look at the effect of a word over three automata simultaneously. So, to convince the reader, we first list the $6$ movements which could increase the number of green states. More precisely, there are only two ways to increase the number of green states:  (1) use a letter $c$ (acting as a contraction) which makes green states appear by saturation, independently of the current state of $A$;  (2) use a letter which moves the current state in $A$ to $m-1$. 

First, we say that a couple of states (on a same automaton) is \emph{uncontractable} if one is green and the other is black. Consider an uncontractable couple over $Q_B$ (or $Q_C$). If we read a word which sends both states on a same state then new green positions could appear (those connected with the black position before reading the word). Regarding this problem, two movements must be avoided:
\begin{itemize}
	\item[(mvt1) -] Reading $c$ when $(q_0,q_1)$ is uncontractable.
	\item[(mvt2) -] Reading $c$ when $(r_0,r_1)$ is uncontractable.
\end{itemize}

Next, we say that a couple of states (on two distinct automata) is \emph{sensitive} if one is green and the other is not green. Consider a sensitive couple of $Q_B\times Q_C$. If we read a word sending it to $(q_0,r_0)$  while we, simultaneously, access to the final state of $Q_A$, then the new couple appearing connects a green position with a non green one. Since, from Lemma \ref{lemmaRT} and by saturation, the green positions always constitute a biclique, the non green state and all those connected to it become green. In relation with this problem, four movements must be avoided:
\begin{itemize}
	\item[(mvt3) -] Reading $a$ when $(q_{n-1},r_0)$ is sensitive and the current state on $Q_A$ is $m-2$.
	\item[(mvt4) -] Reading $b$ when $(q_0,r_0)$ is sensitive and the current state on $Q_A$ is $m-1$.
	\item[(mvt5) -] Reading $c$ when one couple in $\{q_0,q_1\}\times\{r_0,r_1\}$ is sensitive and the current state on $Q_A$ is $m-2$.
	\item[(mvt6) -] Reading  $d$ when $(q_0,r_{p-1})$ is sensitive and the current state on $Q_A$ is $m-1$.
\end{itemize}

Now, we start the induction.
By definition, the state $q_j$ is green. If it is the only one, then two situations may occur:
\begin{itemize}
	\item[i/] If there exists a grey state $r_{k'}\in Q_C$  then the word $d^{p-1-k'}a^{n-1-j}$ 
	does the job, with a first rotation over $Q_C$ to convey the grey state on $r_{p-1}$ and a second one over $Q_B$ to send the green state on $q_{n-1}$. 
	If the reading of the $d$'s from $s_2$ induces (mvt6)-movements (here, this implies $j=0$), we add the prefix $a$ to the previous word and the suffix $a^{n-1-j}$ becomes $a^{n-2}$. 
	\item[ii/] If there is no grey state in $Q_C$ then we create one by making a contraction $c$ which ensures $r_1$ to be immediately grey, and we come back to the previous case. If this contraction from $s_2$ induces a (mvt1)-movement, we first read an $a$ (if $j=1$) or two $a$'s (if $j=0$). Observe that no (mvt3)-movement can appear since there is only one green state (which cannot be, simultaneously, on $q_{n-1}$ and $q_0$ or $q_1$). Similarly, if the contraction from $s_2$ induces a (mvt5)-movement, we read the prefix $aa$ before the $c$.
\end{itemize}

If there are two green states then, necessarily, one is $q_j$ and the other belongs to $R$. We proceed exactly as for case i/ with the second green state playing the role of the grey state $r_{k'}$.\\

Now, suppose there are $\alpha+1$ green states with $\alpha\geq2$. First, observe neither $T$ nor $R$ can be empty. We have two cases to consider.
\begin{itemize}
	\item If $R$ contains at least two states. Among them, we denote $r_{k_1'}$ and $r_{k_2'}$ with $k_2'>k_1'$, those having the biggest indices. We make a rotation to send $r_{k_2'}$ on $r_{p-1}$, then a grouping to have both green positions on $r_0$ and $r_{p-1}$ and last, a contraction to decrease the number of green states  and solve the problem by induction hypothesis. So,  the word $w=d^{p-1-k_2'}(bd)^{k_2'-k_1'}dc$ does the job when $i_2\in Q_A\setminus \{m-2,m-1\}$ . 

If $i_2=m-1$, the word $w$ has to be prefixed with an $a$ to avoid  (mvt4) or (mvt6)-movements.\\
If $i_2=m-2$, 
we  first read some $d$'s depending on the color of $q_{n-1}$.  If it is black or grey, we read $d^{p-k}$ to send $r_k$ on $r_0$, coloring it in black. If it is green, we read $d^{p-k_2'}$ to send $r_{k_2'}$ on $r_0$, coloring it in green. In both cases, we can now read an $a$ from $m-2$ being sure the states $q_{n-1}$ and $r_0$ are not sensitive, and so avoid (mvt3)-movement. Then, we read a last $a$ to reach the state $0$ before reading $w$, where $k'_1$ and $k'_2$ are adjusted in consequence (they are respectively replaced by two new indices $k''_2>k''_1$, where $\{k'_1,k'_2\}$ is sent to $\{k''_1,k''_2\}$ when reading the prefix above).


	\item If $R$ contains only one state then $T$ contains at least two. We denote them $q_{j_1'}$ and $q_{j_2'}$ with $n-1>j_2'>j_1'$ or $j_1'=n-1$ (that is $q_{j_2'}$ is nearest to $q_{n-2}$ than $q_{j_1'}$ when reading only $a$'s) and proceed similarly to the previous case (but not  identically). We make a rotation to send $q_{j_2'}$ on $q_{n-2}$, then a grouping to have both green positions on $q_{n-2}$ and $q_{n-1}$ and last, a contraction to decrease the green states number and solve the problem by induction hypothesis. So, in most of the cases, the word $w=a^{n-2-j_2'}(ab)^{j_2'-j_1'}babbabc$ does the job. The possible problems which can happen are:
	\begin{enumerate}
		\item a (mvt5)-movement if we read the final $c$ from state $m-2$ ;
		\item (mvt3)-movements when reading an $a$ from state $m-2$ ;
		\item (mvt4)-movements when reading a $b$ from state $m-1$ ;
		\item a (mvt2)-movement when reading the final $c$.
	\end{enumerate}
	To solve the first problem, we ensure we are on position $0$ onto the first automaton before reading the suffix $babbabc$. For this, we have to read sufficiently many $ab$ factors in $w$. That is, if we are not on position $0$ after reading the prefix $a^{n-2-j_2'}(ab)^{j_2'-j_1'}$ then we continue to read some $ab$'s: any large enough (for example, larger than $m$) multiple of $n-1$ is a good candidate. Notice that this solution may induce occurrences of problems (2) and (3), which are treated below.

	To solve the second problem, the solution consists, each time we have to read an $a$ in the position $m-2$, to read some $d$'s before the $a$. This allows to color $r_0$ in green or black accordingly to the current color of $q_{n-1}$ to ensure $q_{n-1}$ and $r_0$ to be not sensitive. This is always possible since we have at least one black position (coming from $r_k$) and one green position in $Q_C$ as mentioned above ($R$ is not empty).

	Remark that the reading of some $d$'s produces no  (mvt6)-movement  (because we are on state $m-2$ in $Q_A$). And the adjonction of the $d$'s  has no incidence regarding to the solution we have given for the first problem.

	Observe that the third problem has disappeared since if we are on state $m-1$ when reading a $b$, then we just come from $m-2$ by reading an $a$. As previously mentioned, we first ensured that $q_{n-1}$ and $r_0$ were not sensitive. But this involve, after reading the $a$, that $q_0$ and $r_0$ are now not sensitive, hence avoiding a (mvt4)-movement.

	Last, to avoid the fourth problem, we proceed accordingly to the following (exhaustive) cases:
	\begin{itemize}
		\item If $r_0$ and $r_{p-1}$ are contractable then read $d$ before the final $c$. This action sends $r_0$ to $r_1$ and $r_{p-1}$ to $r_0$ with no move on the other automata. This avoid a (mvt2)-movement while reading the final $c$.
		\item If $r_0$ is uncontractable with $r_{p-1}$, but is contractable with $r_{p-2}$ then read $bdb$ before the final $c$ (the last $b$ is necessary when $n=3$ which implies $q_1=q_{n-2}$).
		This action sends $r_0$ to $r_1$ and $r_{p-2}$ to $r_0$ with no move on the other automata.
		\item If $r_0$ is uncontractable with $r_{p-1}$ and $r_{p-2}$ (that is, the latter two are contractable) then read $dd$ before the final $c$.  
	\end{itemize}

\end{itemize}
\cqfd %
\end{proof}

Therefore,

\begin{theorem}
$\mathrm{sc}(M\cdot(N\oplus P)) = (m-1)\alpha_{n,p}+\alpha'_{n,p}$
\end{theorem}

\bibliography{biblio,bibjg}

\end{document}